\documentclass[11pt]{article}
\usepackage{appendix}
\usepackage{listings}
\usepackage{xcolor}
\usepackage{graphics}
\usepackage{longtable}
\usepackage{indentfirst}
\usepackage{geometry}
\usepackage{hyperref}
\usepackage{amsmath}
\usepackage{amsthm}
\newtheorem{theorem}{Theorem}
\numberwithin{theorem}{section}
\newtheorem{lemma}[theorem]{Lemma}

\newtheorem{corollary}[theorem]{Corollary}
\newtheorem{claim}[theorem]{Claim}

\newcommand{\poly}{{\mathrm{poly}}}

\newcommand{\eps}{\varepsilon}

\newcommand{\mathprob}[1]{\mbox{\textmd{\textsc{#1}}}}

\newcommand{\TQBF}{\mathprob{TQBF}}

\newcommand{\class}[1]{\mathbf{#1}}

\newcommand{\NP}{\class{NP}}
\newcommand{\coNP}{\class{coNP}}
\newcommand{\shP}{\class{\#P}}
\renewcommand{\P}{\class{P}}

\newcommand{\PSPACE}{\class{PSPACE}}

\newcommand{\commentt}[1]{}

\newif\ifdraft
\draftfalse
\ifdraft
\newcommand{\mg}[1]{\textcolor[rgb]{.90,0.00,0.00}{[MG: #1]}}
\newcommand{\mgs}[1]{\marginpar{\tiny \sf \textcolor[rgb]{.90,0.00,0.00}{[MG: #1]}}}
\newcommand{\lgl}[1]{\textcolor[rgb]{.00,0.80,0.00}{[LG: #1]}}
\newcommand{\lgs}[1]{\marginpar{\tiny \sf \textcolor[rgb]{.00,0.80,0.00}{[LG: #1]}}}
\newcommand{\mb}[1]{\textcolor[rgb]{.90,0.50,0.50}{[MB: #1]}}
\newcommand{\mbs}[1]{\marginpar{\tiny \sf \textcolor[rgb]{.90,0.50,0.50}{[MB: #1]}}}
\else
\newcommand{\mg}[1]{}
\newcommand{\mgs}[1]{}
\newcommand{\lgl}[1]{}
\newcommand{\lgs}[1]{}
\newcommand{\mb}[1]{}
\newcommand{\mbs}[1]{}
\fi

\newtheorem{definition}{Definition}[section]

\title{The Complexity of Verifying Boolean Programs
\\ as Differentially Private} 
\date{} 

\author{Mark Bun, Marco Gaboardi, and Ludmila Glinskih \\ Boston University, MA, USA}

\begin{document}
\maketitle
\thispagestyle{plain}
\pagestyle{plain}

\begin{abstract}
We study the complexity of the problem of verifying differential privacy for while-like programs working over boolean values and making probabilistic choices. Programs in this class can be interpreted into finite-state discrete-time Markov Chains (DTMC). We show that the problem of deciding whether a program is differentially private for specific values of the privacy parameters is $\PSPACE$-complete. To show that this problem is in $\PSPACE$, we adapt classical results about computing hitting probabilities for DTMC. To show $\PSPACE$-hardness we use a reduction from the problem of checking whether a program almost surely terminates or not. We also show that the problem of approximating the privacy parameters that a program provides is $\PSPACE$-hard. Moreover, we investigate the complexity of similar problems also for several relaxations of differential privacy: R\'enyi differential privacy, concentrated differential privacy, and truncated concentrated differential privacy. For these notions, we consider gap-versions of the problem of deciding whether a program is private or not and we show that all of them are PSPACE-complete. 
    
\end{abstract}

\section{Introduction}
Differential privacy~\cite{DMNS06} provides a formal framework for guaranteeing that programs respect the privacy of the individuals contributing their data as input. The idea at the heart of differential privacy is to use carefully calibrated random noise to guarantee that an individual's data has a limited influence on the result of a data analysis. The literature on differential privacy shows how this can be done for numerous tasks across statistics, optimization, machine learning, and more. However, showing that a program satisfies differential privacy can be difficult, subtle, and error prone~\cite{Lyu-2017,KiferMRTZ20}. For this reason, several techniques have been proposed in order to verify or find violations in differential privacy programs, e.g.~\cite{ReedP10,barthe2012probabilistic,Gaboardi2013,ZhangK17,DingWWZK18,Bichsel:2018,BCJSV20}.

Despite tremendous progress in the development of methods and tools to support the deployment of differential privacy, there are fundamental open questions about the complexity of the problems these tools address.
In this paper, we focus on one of these problems: 

\begin{quote}
{\bf Approximate-DP}: Given a Boolean 
program and parameters $e^\eps, \delta$, decide whether a program is $(\eps,\delta)$-differentially private or not.
\end{quote}

 Barthe et al.~\cite{BCJSV20} showed that a version of this problem, for probabilistic while-like programs using both finite and infinite data, is undecidable. However, it becomes decidable when a restriction is imposed on the way infinite data are used in while loops. Gaboardi et al.~\cite{GNP20} showed that, for probabilistic programs over finite data domains and without loops, when the parameters are rational, this problem is $\coNP^{\shP}$-complete for $(\eps, 0)$-differential privacy and even harder for $(\eps, \delta)$-differential privacy. In this work we consider the case where programs can contain loops and work over finite data, and the parameters are given as dyadic numbers (rational numbers whose denominator is a power of two). We show that adding loops and maintaining the restriction on finite data preserves decidability but significantly increases the complexity of the problem, even for just $(\eps, 0)$-differential privacy.

\subsection*{Our contributions}
We consider programs from a simple probabilistic while-like programming language over boolean data, where  randomness is represented as probabilistic choice. We call this language BPWhile. This language can be seen as a low-level target language for differential privacy implementations which are intrinsically over finite data types~\cite{Mironov12,GazeauMP16,BalcerV18,Ilvento20}. 

As a first step, we show $\PSPACE$-hardness for {\bf Approximate-DP} over this language, with respect to the size of the program. We show this result by using a reduction from the problem of deciding \emph{almost sure termination} for programs in BPWhile. Programs in this language can be seen as discrete-time recursive Markov chains for which almost sure termination has been shown $\PSPACE$-complete~\cite{EtessamiY09}.
Intuitively, the hardness of verifying whether a program is differentially private comes from the fact that we need to compare distributions on outputs for neighboring pairs of inputs. Understanding such distributions essentially gives us a way to check whether a program terminates with probability 1 or not. We use this idea in all $\PSPACE$-hardness proofs in this work.

We then present an algorithm for {\bf Approximate-DP} which uses polynomial space, completing our proof of $\PSPACE$-completeness for {\bf Approximate-DP}. Our algorithm is based on  classical results showing that computing hitting probabilities in discrete-time Markov chains can be done in a space efficient way.
Our proof of $\PSPACE$-completeness even holds in the case where the privacy parameter $\delta$ is zero---this setting is usually called \emph{pure differential privacy}. 

Similarly to~\cite{GNP20}, we also consider a related problem concerning the approximation of privacy  parameters. In particular, we study the following gap-promise variant of the problem.
\begin{quote}
{\bf Distinguish $(\eps,\delta)$-DP}: Given a program that is promised to either be $(0,0)$-differentially private or not $(\eps,\delta)$-differentially private, decide which is the case. Here, $\eps, \delta$ may be fixed constants independent of the input.
\end{quote}
We show that this problem is also $\PSPACE$-hard via another reduction from the problem of deciding almost sure termination. At first glance, the statement seems specific to $(0,0)$-differentially privacy, but it implies more generally that it 
is hard to distinguish between $(\eps, \delta)$-differentially private programs and programs which fail to be $(\eps + \alpha, \delta + \beta)$-differentially private for positive constants $\alpha, \beta$.
In particular, it is hard even to approximate the best $\eps$ and $\delta$ parameters for which a program guarantees differential privacy. 

Further, we consider several relaxations of the  definition of differential privacy which have recently appeared in the literature. Specifically, we consider deciding R\'enyi-differential privacy (RDP)~\cite{Mironov17}, concentrated differential privacy (CDP)~\cite{BunSteinke16}, and truncated concentrated differential privacy (tCDP)~\cite{BDRS18}. For each of these privacy notions we define a gap version of the problem of deciding whether a program is private or not and we show that all of them are $\PSPACE$-complete.

To show membership in $\PSPACE$ we use similar approach as in the $\PSPACE$-algorithm for {\bf Decide $(\eps, \delta)$}-DP. The main difference is that definitions of RDP, CDP, and tCDP involve computations of R\'enyi divergences and, as we are working with probabilities that can have exponentially long descriptions, we carefully apply known uniform families of polylogarithmic depth circuits to perform these calculations. We prove our lower bounds using reductions from {\bf Distinguish $(\eps,\delta)$-DP}.

To summarize, our contributions are: 

\begin{enumerate}
    \item We give a proof of $\PSPACE$-hardness for the problem of deciding $(\eps,\delta)$-differential privacy (by showing a polynomial time reduction from the language of almost surely terminating programs, Section \ref{sec:dp-ver-pspace-hard}).
    \item We show a $\PSPACE$ algorithm for deciding $(\eps,\delta)$-differential privacy (Section \ref{pspace-approx-dp-algo}).
    \item We show $\PSPACE$-hardness for the problem of approximating the privacy parameters (Section \ref{sec:approx-dp-param-pspace-hard}).
    \item We show $\PSPACE$ algorithms for deciding R\'enyi-differential privacy (Section \ref{sec:pspace-rdp-algo}), concentrated differential privacy (Section \ref{sec:pspace-cdp-algo}), and truncated differential privacy (Section \ref{sec:pspace-tcdp-algo}).
    \item We also give a proof of $\PSPACE$-hardness for deciding R\'enyi-differential privacy (Theorem~\ref{th:RDP-pspace-hard}), concentrated differential privacy (Theorem~\ref{th:CDP-pspace-hard}), and truncated concentrated differential privacy (Theorem~\ref{th:TCDP-pspace-hard}) (via reductions from the problem of approximating privacy parameters, Section \ref{sec:approx-dp-param-pspace-hard}).
\end{enumerate}

\section{Related work}
\paragraph*{Verification tools for differential privacy.}
Several tools have been developed with the goal of supporting programmers in their effort to write code that is guaranteed to be differentially private, including type systems~\cite{ReedP10,Gaboardi2013,BartheGAHRS15,ZhangK17,NearDASGWSZSSS19}, program logics~\cite{barthe2012probabilistic,BartheGAHKS14,Barthe:2016}, and other program analyzers~\cite{tschantz2011formal,fredrikson2014satisfiability,albarghouthi2017synthesizing,LiuWZ18,chistikov2018bisimilarity}. Other tools help programmers find violations in differentially private implementations~\cite{DingWWZK18,Bichsel:2018,ZhangRHP020}. Finally, several recent tools address both  problems at the same time~\cite{WangDKZ20,BCJSV20,FarinaCG20}. Most of  these tools are capable of analyzing complex examples corresponding to the state of the art in differential privacy algorithm design~\cite{DingWZK19,KaplanMS20}.

\paragraph*{Implementations on finite computers.}
Several works have studied how to implement differentially private algorithms using finite arithmetics. 
Mironov~\cite{Mironov12} showed that na\"ive implementations of the Laplace distribution using floating point numbers are actually not private. Gazeau et al.~\cite{GazeauMP16} showed that similar problems as the one identified by Mironov are not only due to the non-uniformity of floating points but they are actually intrinsically due to the use of finite precision arithmetic. Ilvento~\cite{Ilvento20} showed that similar considerations can be applied also to algorithms that are in principle discrete, such as the exponential mechanism. Balcer and Vadhan~\cite{BalcerV18} showed how to implement several important differentially private algorithms in an efficient way on finite precision machines.  
\paragraph*{Related results in complexity.}
Murtagh and Vadhan~\cite{MurtaghV16} studied the complexity of finding the best privacy parameters for the composition of multiple differentially private mechanisms and showed it to be $\shP$-complete. This work, in part, led to the development of several variants of differential privacy, most of which we consider here, with better composition properties. 
Barthe et al.~\cite{BCJSV20} showed that deciding differential privacy for probabilistic while-like programs using both finite and infinite data is undecidable, but it becomes decidable when a restriction is imposed on the way infinite data are used in while loops. However, they do not study the computational complexity of this problem. 
Gaboardi et al.~\cite{GNP20} showed $\coNP^{\shP}$-completeness for the problem of deciding $(\eps, 0)$-differential privacy  for probabilistic programs over finite data domains and without loops. They also studied this problem and approximate versions of it for $(\eps, \delta)$-differential privacy. 
Chadha et al.~\cite{ChadhaSV21} recently showed that deciding differential privacy for a class of automata that can be used to describe classical examples from the differential privacy literature can be done in linear time in the size of the automata. This class of automata includes computations over unbounded input data, such as real numbers.
Chistikov et al.~\cite{chistikov2018bisimilarity,chistikov2019asymmetric} studied several complexity problems concerning differential privacy in the setting of labeled Markov chains. They showed that the threshold problem for a computable bisimilarity distance giving a sound technique to reason about differential privacy is in $\NP$~\cite{chistikov2018bisimilarity}. Further, they proved that another distance, based on total variation,  which can be used to more precisely reason about differential privacy is undecidable in general, and the problem of approximating it is $\shP$-hard, and in $\PSPACE$~\cite{chistikov2019asymmetric}.

There are also other related results from the program verification and privacy literatures. 
Courcoubetis and Yannakakis~\cite{CourcoubetisYannakakis95} studied the complexity of several verification problems for probabilistic programs. 
Etessami and Yannakakis~\cite{EtessamiY09} studied the complexity of several problems for recursive Markov chains. Notably, they showed that deciding almost sure termination for this computational model is $\PSPACE$-complete.
Kaminski et al.~\cite{KaminskiKM19} studied the arithmetic complexity of almost sure termination for general probabilistic programs with unbounded data types. 
Chadha et al.~\cite{ChadhaKV14} showed $\PSPACE$-completeness for the problem of bounding quantitative information flow for boolean programs with loops and probabilistic choice. A bound on pure differential privacy entails a bound on quantitative information flow, but not the other way around, and hence their result does not directly apply in our context. 
Gilbert and McMillan~\cite{GilbertM18} studied the query complexity of verifying differential privacy programs modeled as black boxes.

\section{Preliminaries}
\subsection{Boolean Programs with Loops and Random Assignments}
In this paper we consider a simple while-like language working over booleans, extended with probabilistic choice. This language, which we call BPWhile, can be seen as a probabilistic extension of the language for input/output bounded boolean programs studied in~\cite{GY13}.
The syntax of the language is defined by the following grammar. 
\begin{eqnarray*}
    b & ::= & {\tt true} \mid {\tt false} \mid {\tt random} \mid x \mid b \land b \mid b \lor b \mid \: !b \\[-1mm]
    c &::=& {\tt skip} \mid x := b \mid c; c \mid {\tt if }\:  b\:  { \tt then}\:  c \: { \tt else }\: c \mid {\tt while}\: b\: {\tt then}\: c \\[-1mm] 
    C &::=& {\tt input}(x,\ldots, x); c; {\tt return}(x,\ldots, x) 
\end{eqnarray*}

All of the constructs are standard. The expression $\tt random$ represents a random fair coin, which with probability 1/2 evaluates to true and with probability 1/2 evaluates to false. 
The semantics for BPWhile programs is also standard and we omit it here. However, notice that program may fail to terminate, and we also have to consider this when analyzing probabilities. To mark non-termination we will use the symbol $\bot$. We also remark that a given BPWhile program operates only on boolean inputs of a single fixed length $n$, specified (implicitly) in the program description.
 
Our language is very similar to the one studied in \cite{GNP20}. The main difference is that we have an additional loop construction \textbf{while $b$ then $c$}.
Without loops, programs in this language can be interpreted into boolean circuits of roughly the same size. However, this cannot be done in presence of loops, as the straightforward approach of unfolding loops gives a circuit of size exponential in the program length. To avoid analyzing boolean circuits of exponential size, we will instead analyze programs as discrete-time Markov chains, in a manner similar to~\cite{BCJSV20}. This is possible because BPWhile programs use a bounded amount of memory (that is at most linear in the size of the input program), corresponding to an exponential, in the size of the input, number of states in the resulting Markov chain. The precise translation will be given in Theorem~\ref{thm:pure-dp-alg}.

Similarly to  \cite{GNP20} we measure the complexity of the problems we are interested in as functions of the size of the input program, rather than, e.g., the number of bits the input program itself takes as input.

\paragraph*{Language expressivity.}
We use booleans as our basic data type to keep our proofs simple. However,  all of the results we show also hold for programs where values are from a fixed finite domain. 
In fact, the language we use here can be thought as a low-level language which could be the target of implementations of differential privacy primitives. As shown in several works, one has to be very careful when implementing differentially private primitives~\cite{Mironov12,GazeauMP16,BalcerV18,Ilvento20}. One way to guarantee correctness for this process could be to give a translation into BPWhile and then decide whether the given program is differentially private or not. We illustrate how this process could work with an example. 

Using $1+n+m$ boolean values we can represent arbitrary positive and negative fixed-point numbers with range $(-2^{n}+1,2^{n}-1)$ and precision $2^{-m}$, and perform standard arithmetic operations and comparison over them. We can then think about working with blocks of variables of size $1+n+m$, which we denote using vector notation, e.g. $\vec{x},\vec{y},\ldots$. Notice that using this representation we can also easily encode a uniform sampling operation for elements in a range (v,w], which we denote ${\tt uniform}(v,w]$. 
We can, for example, implement the bounded Geometric Mechanism from~\cite{GhoshRS12}, using this encoding and the implementation in finite precision arithmetic provided in~\cite{BalcerV18}. 
Given a positive integer $n$ and a private positive integer value $c\leq n$, this discrete mechanism  selects an integer element $z$ from the range $[0,n]$ with probability proportional to $e^{\frac{-\eps|z-c|}{2}}$. Essentially, the mechanism implements inverse transform sampling based on the inverse CDF of the output distribution. Given $c$, $n$ and $\eps$, this mechanism can be described in BPWhile as in Figure~\ref{fig:geom-mech}.
\begin{figure}
    \centering
\begin{tabular}{ll}
0.& ${\tt input}(\vec{c},\eps);$\\
1.& $\vec{k}:=\lceil \log(2/\eps)\rceil;$\\
2.& $\vec{d}:=(2^{\vec{k}+1}+1)(2^{\vec{k}}+1)^{n-1};$\\
3.& $\vec{u}:={\tt uniform}(0,\vec{d}];$\\
4.& $\vec{z}:=0;$\\
5.& $\vec{r}:=n;$\\
6.& ${\tt while}\, \vec{z}< \vec{n} \land \vec{r}=n \, {\tt then}$\\
7.& $\quad  {\tt if}\ \vec{z} < \vec{c}\ {\tt then}$\\
8.& $\quad \quad {\tt if}\ \vec{u} \leq 2^{\vec{k}(\vec{c}-\vec{z})}(2^{\vec{k}}+1)^{n-(\vec{c}-\vec{z})} $\\
9.& $\quad \quad {\tt then}\ \vec{r}:=\vec{z}$\\
10.& $\quad \quad {\tt else}\ {\tt skip}$\\
11.& $\quad  {\tt else}$\\
12.& $\quad \quad {\tt if}\ \vec{u} \leq d-2^{\vec{k}(\vec{z}-\vec{c}+1)}(2^{\vec{k}}+1)^{n-1-(\vec{z}-\vec{c})} $\\
13.& $\quad \quad {\tt then}\ \vec{r}:=\vec{z}$\\
14. & $\quad \quad{\tt else}\ {\tt skip}$\\
15.& $\quad  \vec{z}=\vec{z}+1;$\\
16.& ${\tt return}(\vec{z});$\\
\end{tabular}
    \caption{Example: Bounded Geometric Mechanism in finite precision arithmetic}
    \label{fig:geom-mech}
\end{figure}

All the operations in this piece of code are assumed to work on blocks of variables and of booleans that are long enough to avoid overflow and approximations.  
This algorithm samples from a uniform distribution (line 3) for a value of $\vec{d}$ large enough and uses a while loop to go through the integers in the range $[0,n]$ to find the right element to return. A faster implementation could be based on binary search. The nested conditionals (lines 7-14) implement the checks required for the inverse transform sampling to identify the right element to return. 

We gave this example to show that the language is expressive enough to implement a real-world mechanism. However, we also chose this example because identifying the privacy guarantee provided by this algorithm is non-trivial. Balcer and Vadhan~\cite{BalcerV18} showed this algorithm to be $(\tilde{\eps},0)$-differentially private when $\tilde{\eps}=\ln(1+2^{-\lceil \log(2/\eps)\rceil})$ and $\tilde{\eps}\in (2/9\eps,\eps/2]$, where the complexity in the expression for $\tilde{\eps}$ comes from the implementation. This example shows why several works have designed methods to decide differential privacy, and why it is important to understand the complexity of this problem.

\subsection{Almost Sure Termination and Configuration Graph}
Our approach will rely on the hardness of the problem of deciding {\em almost sure termination} for probabilistic boolean programs (Lemma \ref{losslessness_pspace_hard}). Almost sure termination is a natural probabilistic extension of the concept of termination.

\begin{definition}
A program $C$ \emph{almost surely terminates} if on all inputs it terminates with probability 1.
\end{definition}
Deciding almost sure termination for general probabilistic programs on unbounded data types is known to be $\Pi_2^0$-complete~\cite{KaminskiKM19} while for programs representing recursive Markov chains it is known to be $\PSPACE$-complete~\cite{EtessamiY09}.

In the following, it will be convenient to analyze BPWhile programs using their configuration graph. To do this, we assume that the code of a program comes with lines of code associated to each command, in a way similar to the code in Figure~\ref{fig:geom-mech}.

\begin{definition}
Consider a BPWhile program $C$ with $l$ lines of code and $v$ Boolean variables. A \emph{state} $s$ of $C$  is a pair $(m,i)$ where  $m \in \{0,1\}^v$ represents a potential value of the memory, i.e. values for all the variables, and $i\in [l]$  is a line of code. The \emph{size} of the program is the number of symbols in the description of the program.
\end{definition}

Note that as the description of each variable, input value and line in the program requires at least one symbol, we get that $l$, $v$, and the size of the input of the program are always at most the size of the program. Throughout this paper we measure complexity of the verifying procedures based on the size of the program.
 
\begin{definition}
The configuration graph $G=(V,E)$ of a BPWhile program $C$ on input $x$ has a vertex for every possible state of the program and a directed edge $((m,i),(m',i'))\in E$ if the  probability of getting the memory $m'$ starting from the memory $m$ and executing the command at line $i'$  is strictly greater than $0$.
\end{definition}
We will also sometimes use the term \emph{state graph} to refer to the configuration graph. 
The starting state of a program is the state at the beginning of $C$'s execution on $x$, where the input variables are set to $x$ and the index of the execution line is $0$. A final state is any state following the execution of the last line of code (the final return command). We denote the set of final states in a configuration graph by $V_f \subseteq V$.

\subsection{Markov Chains}
To analyze the probability that a boolean program $C$ on input $x$ outputs a specific value, we need to associate probabilities to each transition in the configuration graph. By doing this, we turn a configuration graph into a discrete time Markov chain.
\begin{definition}[\cite{HartSharir84,LehmannShelah82,Vardi85}]
A discrete-time Markov chain $M=(V,E,\{p_{uw} \mid (u, w) \in E\},\{p_0(v) | v\in V\})$ consists of a set of states $V$, a set $E \subseteq V \times V$ of transitions between states, a list $p_{uw}$ of positive probabilities for all transitions $(u,w)$ such that for each state $u \in V$, we have $\sum_{w\in V}{p_{uw}}=1$, and an initial probability distribution $p_0$ on states in $V$.
\end{definition}

Following \cite{CourcoubetisYannakakis95} we view a Markov chain as a directed graph $(V, E)$, with weights $p_{uw}$ on all edges $(u,w)$. Moreover, as in a configuration graph, we associate each vertex in the graph to a state of a BPWhile program, and a transition between states to one possible execution step of the program. As an initial probability distribution we use a unit distribution that places weight $1$ on the unique start state of the program.

To verify whether a BPWhile program is differentially private, as we will see in the next section, we need to compare the probabilities of outputting the same output on neighboring inputs. We will do this by computing hitting probabilities of final states with a fixed output values.
\begin{definition}
The \textit{hitting probability} of a state $s \in V$ in a Markov chain $M=(V,E,p,p_0)$ is the probability of reaching $s$ in $M$ starting with a initial probability distribution $p_0$ after an arbitrary number of steps.
\end{definition}

\subsection{Differential Privacy}

Differential privacy is a property of a program that can be expressed in terms of a neighboring relation over possible program inputs. Here we view an input as a sensitive dataset, and say that two inputs are neighboring if they differ in one individual's information.
As our focus in this paper is on boolean programs, we define two datasets to be neighboring when they differ in a single bit.
\begin{definition}
Two boolean vectors of the same length are said to be \emph{neighboring} if their Hamming distance (the number of positions in which these vectors differ) equals $1$.
\end{definition}

Notice that this is a strong notion of neighboring, which makes our hardness results stronger. That is, our hardness results extend naturally to other more involved notions of neighboring.
Moreover, our upper bound arguments apply to any neighboring relation between boolean vectors (or more generally, vectors over any fixed finite data domain) as long as that relation can be checked in polynomial space. Using the notion of neighboring we introduced above we can now formulate  differential privacy.

Differential privacy guarantees that a change of any one data in the input will not change much the observed output of the program. More formally, differential privacy guarantees that the distributions of outputs of a program when run on neighboring datasets are close.

\begin{definition}[Differential Privacy~\cite{DMNS06}]
\label{def:differential-privacy}
A boolean program $C$ with inputs of length $n$ and producing outputs of length $l$ is $(\eps, \delta)$-differentially private if for every pair of neighboring inputs $x,x'\in \{0,1\}^n$ and for every set of possible outputs $O \subseteq \{0,1\}^l \cup \{\bot\}:$ 
\begin{equation} \label{eqn:dp-ineq}
\Pr[C(x) \in O] \leq e^\eps \cdot \Pr[C(x') \in O] + \delta.
\end{equation}
\end{definition}

This version of differential privacy is often called \emph{approximate differential privacy} to distinguish it from \emph{pure} differential privacy, which is the special case where $\delta=0$. We will denote the latter by $\eps$-differential privacy. 

In the following, it will be convenient at times to work with the following reformulation of differential privacy. 
\begin{lemma}
[Pointwise differential privacy \cite{Barthe:2016}]\label{pointwiseDP}
A program $C$ is $(\eps,\delta)$-differentially private if and only if for all neighboring inputs $x, x' \in \{0,1\}^n$,
$$\sum\limits_{o \in \{0,1\}^l \cup \{\bot\}}{\delta_{x,x'}(o)} \leq \delta,$$
where $\delta_{x,x'}(o)=\max(\Pr[C(x)=o]-e^\eps(\Pr[C(x')=o]),0)$.
\end{lemma}

\section{Complexity of Checking Almost Sure Termination}
In this section we give intuition for the hardness of deciding differential privacy by discussing the complexity of almost sure termination. While it is known that almost sure termination for Markov chains is $\PSPACE$-complete~\cite{EtessamiY09}, we believe it is instructive to understand where this complexity comes from. We start with a helpful characterization of almost sure termination of a BPWhile program in terms of reachability in  the program's configuration graph.

\begin{theorem}
\label{losslessness-criterion}
A program $C$ terminates almost surely if and only if for every input $x$ and every vertex $v$ in the configuration graph of $C(x)$ that is reachable from the start state, there is a path from $v$ to one of the final states.
\end{theorem}
\begin{proof}
For the ``if'' direction, suppose $x$ is an input such that for every reachable vertex in the state graph $G = (V, E)$ of $C(x)$, there is a path from $v$ to one of the final states. Let $m = |V|$ be the number of vertices. Since every simple path in $G$ has at most $m$ edges, we have that for every $v$, the probability of reaching a final state after at most $m$ additional steps of computation starting from $v$ is least $2^{-m}$. Therefore, for any $k \ge 1$, the probability that the program fails to terminate on input $x$ after $km$ steps is at most $(1-2^{-m})^k$. Taking $k \to \infty$, we see that that the program fails to terminate with probability $0$. Therefore, $C$ terminates almost surely on input $x$.

For the ``only if'' direction, suppose there is an input $x$ and a vertex $v$ in the state graph of $C(x)$ that is reachable from the start state but cannot reach any final state. Then on $C(x)$ reaches state $v$ with probability at least $2^{-m}$ by following the simple path from the start state to $v$. Once the program has reached $v$, it is impossible to terminate. So the program terminates with probability at most $1-2^{-m} < 1$.
\end{proof}

The main intuition of this theorem is that the only way for a program to fail to terminate with probability $1$ is if there is some positive probability that it enters an infinite loop from which it cannot exit. This is possible if and only if there exists a state that is reachable from the start state, but from which we cannot reach any of the final states.

Theorem \ref{losslessness-criterion} immediately suggests a simple exponential-time (and exponential-space) algorithm for checking almost sure termination. For each possible input to the program, we can construct the configuration graph of the program on that input. Using breadth-first search, we can mark which states are reachable from the start state, and for each such state we check whether any of the final states are reachable. If there exists an input and a state in its configuration graph that is reachable from the start state but cannot reach a final state, then by Theorem \ref{losslessness-criterion} we get that the program doesn't terminate almost surely.  If for every input, there is no such state, then the program almost surely terminates.

Constructing the configuration graph explicitly and running breadth-first search uses exponential space. In what follows, we describe how to reduce the space complexity to polynomial.

We can improve the previous algorithm by avoiding storing the whole configuration graph, and instead  providing implicit access to any edge in the graph.

This requires us to re-compute on-the-fly information about the set of reachable states from any given vertex, but fortunately, this can still be done in polynomial space.

\begin{theorem}
There is a deterministic algorithm for checking almost sure termination of a BPWhile program using space polynomial in the size of the program. 
\end{theorem}

To show $\PSPACE$-hardness of checking almost sure termination we reduce from the  $\PSPACE$-complete true quantified boolean formula ($\TQBF$) problem. This is the problem of deciding whether a fully quantified propositional boolean formula is true. For a formula $\phi$ with $t$ quantifiers, we define a BPWhile program with $t$ nested while-loops to evaluate the formula. The reduction is similar to the reduction in \cite{GY13} from $\TQBF$ to the reachability problem for extended hierarchical state machines. We prove the following theorem in Appendix \ref{TQBF_to_losslessness}.

\begin{theorem}\label{losslessness_pspace_hard}
The  problem of checking whether a BPWhile terminates almost surely is $\PSPACE$-hard.
\end{theorem}

\section{$\PSPACE$-Completeness for Pure and Approximate Differential Privacy}
We reason about differential privacy in a manner similar to almost sure termination. In particular, we  use a Markov chain interpretation of a program $C$. 

We first give a formal definition of the verification problems we consider. Then we give an inefficient (exponential-time) but simple algorithm (Section \ref{expspace-pure-dp-algo}), followed by a $\PSPACE$-algorithm for verifying whether a program is pure (Section \ref{pspace-pure-dp-algo})  or approximate differentially private (Section \ref{pspace-approx-dp-algo}).

\begin{definition}
In the {\sc Pure-DP} problem, an instance $(C, e^\eps)$ consists of a BPWhile program $C$, and a dyadic rational number $e^\eps$. The problem is to distinguish whether for all neighboring inputs $x,x'$ and for every set of possible outputs $O$ we have 
$$\Pr[C(x) \in O] \leq e^{\eps} \cdot \Pr[C(x) \in O].$$ 
\end{definition}

\begin{definition}
In the {\sc Approximate-DP} problem, an instance $(C, e^\eps, \delta)$ consists of a BPWhile program $C$, and two dyadic rational numbers $e^\eps, \delta$. The problem is to distinguish whether for all neighboring inputs $x,x'$ and for every set of possible outputs $O$ we have 
$$\Pr[C(x) \in O] \leq e^{\eps} \cdot \Pr[C(x) \in O] + \delta.$$ 
\end{definition}

\subsection{Exponential-Time Algorithm for Checking $(\eps, 0)$-Differential Privacy}\label{expspace-pure-dp-algo}
To give an exponential-time algorithm for checking $(\eps, 0)$-differential privacy, we first review  the algorithm for computing the probability of reaching any given final state $s_f$ in a Markov chain  from \cite{BCJSV20}:
\begin{enumerate}
     \item For each state $v$, initialize a variable $q_v$ representing the probability of reaching $s_f$ from this state.
     \item Set $q_{s_f} = 1$ for the final state $s_f$.
     \item For each state $v$ from which $s_f$ is not reachable set $q_v = 0$.
     \item For any state $v$ for which we do not yet have an equation, we introduce the equation $q_v=\sum_{u \in V} q_u \cdot p_{v u} $, where $p_{vu}$ is the probability of transitioning from $v$ to $u$ in one step. If there is no transition from $v$ to $u$, then $p_{vu} =0$.  
     \item The previous steps give us a set of equations, one for each possible state of the Markov chain. The number of variables equals the number of equations. This linear system can be solved unambiguously by using any polynomial-time algorithm for solving systems of linear equations. 
 \end{enumerate}
We can now state our exponential time algorithm for deciding differential privacy for BPWhile programs. 
 \begin{theorem}
 {\sc Pure-DP} problem is solvable by a deterministic algorithm using time exponential in the size of a program.
 \end{theorem}
 \begin{proof}
By Lemma~\ref{pointwiseDP} for checking $(\eps, 0)$-Differential Privacy for a program $C$ it is sufficient to compare for every pair of neighboring dataset the output distributions on every possible value. Using this approach, we get the following simple algorithm:
     \begin{enumerate}
         \item For neighboring inputs $x,x'$ of size $n$ and a program $C$ of size $N$ construct two Markov chains, with one start-state in each, set these start-states to $x$ and $x'$, respectively.
         \item Find the probabilities of each final state in each Markov chain.
         \item Compare the probability of the same states in the two Markov chains. If there is at least one output $c$ such that $P[C(x)=c] > e^{\eps} P[C(x')=c]$, terminate and output ``Not $(\eps, 0)$-DP''. Otherwise continue. 
         \item If the checks were successful for all pairs, terminate with an output: ``$(\eps, 0)$-DP''.
     \end{enumerate}
 This algorithm explicitly store probabilities of reaching all of up to $2^{N}$ final states, as well as a system of linear equations of size exponential in $N$. As the input of the algorithm is a program $C$ of size $N$, we get that this algorithm requires exponential space and time in the size of its input.
\end{proof}

\commentt{
\textbf{Time and space analysis:} this algorithm requires exponential space and it solve a system of linear equations on exponential number of variables. As algorithm for solving such system takes polynomial number of steps in a number of variables (even the simplest one, like Gauss's method) overall this algorithm requires exponential time.

In case of $(\eps, 0)$-Differential Privacy it is sufficient to compare a probability to get the same output for each pair of inputs.

Our final algorithm looks as follows:
\begin{enumerate}
    \item For each pair of neighboring inputs $(a,b)$:
    \begin{enumerate}
        \item Construct two Markov chains, with one start-state in each, set these start-states to correspond to a memory set to $a$ and $b$.
        \item Find a probability of each final state in each Markov chain.
        \item Compare probability of the same states for both Markov chains, if there at least one output $c$ such that $P[C(a)=c] > e^{\eps} P[C(b)=c]$ stop the procedure and output "Not $(\eps, 0)$-DP". Else continue. 
    \end{enumerate}
    \item If checked for all pairs, terminate with an output: "$(\eps, 0)$-DP".
\end{enumerate}

\textbf{Time and space analysis:} we have an exponential number of different inputs, for each pair we construct a Markov chain of exponential size, and then run a procedure that finds and compares an exponential number of probabilities in exponential time. Overall this algorithms uses exponential amount of space and time.}

\subsection{PSPACE Algorithm for Checking $(\eps, 0)$-Differential Privacy}\label{pspace-pure-dp-algo}

A classic line of work~\cite{Sim81, BorodinCookPippenger83, Jung81} showed that computing the hitting probabilities of final states can be done efficiently in space. This is what we need to design a $\PSPACE$ algorithm to check differential privacy. In designing this algorithm we use the work of Simon~\cite{Sim81} who showed that given a Markov chain of size $M$, the hitting probability of any state can be computed in space $O((\log{M})^6)$. Subsequent work~\cite{BorodinCookPippenger83, Jung81} improved this result by showing that $O((\log{M})^2)$ is enough. Nevertheless, we focus our exposition on Simon's algorithm as its presentation simplifies the description of our algorithm, and improving the polynomial does not affect membership of our problem in $\PSPACE$.

Simon's result can be formally stated as follow:

\begin{lemma}\cite{Sim81}\label{hitting_prob}
Let $M$ be a Markov chain (represented by its transition matrix) with at most $2^{ L}$ states, an initial distribution placing all mass on one state $s$, a set of final states $F$ each with only one self-transition, and the property that every state not in $F$ each outgoing transition probability is either $0$ or $1/2$. There is an $O(L^6)$-space deterministic algorithm that computes the hitting probabilities of every state in $F$.
\end{lemma}

To apply the algorithm from the previous lemma we need to do an extra pre-processing step to remove all non-final recurrent states of a Markov chain.

\begin{definition}
A \textit{recurrent state} in a Markov chain is a state such that, after reaching it once, the probability of reaching it again is $1$. 
\end{definition}

A similar pre-processing step appears in Simon's paper, and we describe our removal process below in our proof of Theorem~\ref{thm:pure-dp-alg}.

Now we are ready to show that  $(\eps,0)$-differential privacy for BPWhile programs can be decided in polynomial space. 

\begin{theorem} \label{thm:pure-dp-alg}
The {\sc Pure-DP} problem is solvable by a deterministic algorithm using space polynomial in the size of the program.
\end{theorem}
\begin{proof}

To apply the algorithm from Lemma~\ref{hitting_prob} and conclude that polynomial space is sufficient in order to compute the final probabilities, we need to be able to compute the probability of each transition in the Markov chain using polynomial space. We cannot explicitly store the whole Markov chain using space that is polynomial in the size of a program. Instead, we can construct an algorithm working in polynomial space which gets as input a description of the BPWhile program $C$, the program input $x$, and two states $u,v$ of the  Markov chain corresponding to $C(x)$. It outputs the transition probability of edge $(u,v)$ (the probability that $C(x)$ gets from state $u$ to state $v$ in one step).

We need to find the probability of hitting each reachable final state of the Markov chain of $C(x)$. Note that these probabilities can be as small as $1/2^{2^{p(N)}}$ for some polynomial $p(N)$, where $N$ is the size of the input program. This is because a Markov chain for a program of size $N$ has a number of states which is at most exponential in $N$, and as each transition probability is either 0, or 1/2, or 1, there can be a simple path in the Markov chain from the start state to the final state that goes through all the states with probability $1/2^{2^{p(N)}}$.

Storing these values requires exponential space, so the $\PSPACE$ algorithm described further only provides implicit access to these probabilities, i.e., the ability to compute any desired bit of a probability.

Here are the conditions that the Markov chain we construct needs to satisfy in order to apply Lemma~\ref{hitting_prob}:
\begin{itemize}
    \item The transition probability between every two states in the Markov chain of size $O(2^{\poly(N)})$ should be computable in polynomial space. Every final state has a self-transition with probability $1$. 
    \item Each transition in the Markov chain for all non-final states has to have weight either $1/2$ or $0$, and the graph underlying the Markov chain shouldn't contain multiple edges. This can be done by duplicating every state, except the start state, increasing the number of vertices by a factor of 2. Every duplicate final state is also marked as a final state.  Let $a$ and $b$ be vertices in the original Markov chain of the program that are transformed to two pairs of vertices $a_1, a_2$ and $b_1, b_2$ respectively. Then we re-assign the weight of edge $e$ from $a$ to $b$ in the original Markov chain as follows:
    \begin{itemize}
        \item If the original weight of $e$ is $1/2$, then we add two edges $(a_1, b_1), (a_2, b_2)$ each of weight $1/2$ to the new Markov chain.
        \item If the weight of $e$ is $1$ we add four edges $(a_1, b_1), (a_1,  b_2), (a_2, b_2), (a_2, b_1)$ each of weight $1/2$. 
        \item If the weight of $e$ is 0, we do not add any edges between vertices $a_1, a_2$ and $b_1, b_2$.
    \end{itemize}
    Therefore, for each original edge we add at most 4 new edges, so we do not increase the size of the Markov chain by more than a factor of 4. Moreover, for every pair of vertices in the new Markov chain, we can recompute the weight of the edge based on the the weight of the edge in the original Markov chain in linear time. Overall, this transformation is computable in the space polynomial in the size of the input BPWhile program and it guarantees that the probability of getting from one vertex to any other in one step is either 0 or $1/2$.

    \item All recurrent states except the final states should be deleted. We simulate this deletion as follows. Whenever our algorithm reads the probability on an edge $(u, v)$, we check whether either $u$ or $v$ are recurrent and zero out this probability if so. This check is similar to the one that we discussed earlier for almost sure termination. We consider the graph underlying the Markov chain of the program. To check whether state $u$ is recurrent, we run a search algorithm checking whether there is at least one path through edges with non-zero weight to at least one of the final states. We can use Savitch's algorithm \cite{Savitch70} to do this check in space polynomial in the size of the program.
\end{itemize}
    
To verify whether a program $C$ is $\eps$-differentially private we can now enumerate all pairs $x, x'$ of neighboring inputs, and all possible outcomes $o$. For each outcome $o$, we compute $\Pr[C(x) = o]$ by summing the hitting probabilities of reaching final states in the configuration graph of $C$ on $x$ that result in outputting $o$. Finally, we compare $\Pr[C(x) = o]$ to $e^{\eps}\Pr[C(x') = o]$. Note that if $e^\eps$ is a rational number with numerator $a$ and denominator $b$, then we can avoid division by comparing $b \cdot \Pr[C(x) = o]$ to $a \cdot \Pr[C(x') = o]$. 

We remark that the necessary arithmetic operations on exponentially long (implicitly represented) numbers can be carried out in polynomial space (though exponential time is still required) using classic logspace algorithms for addition and multiplication.\footnote{We can construct uniform $NC^1$ and $NC^2$ circuits for these operations. Simple constructions are described in \cite{Wegener87}.} In particular, this works even if $e^\eps$ is an exponentially long rational number provided as input to the problem.
\end{proof}

\subsection{PSPACE Algorithm for Checking $(\eps, \delta)$-Differential Privacy}\label{pspace-approx-dp-algo}
Now, using the pointwise definition of differential privacy from Lemma \ref{pointwiseDP} and using similar ideas to the algorithm in Section \ref{pspace-pure-dp-algo} we can construct a $\PSPACE$-algorithm for checking $(\eps, \delta)$-differential privacy of BPWhile programs.

\begin{theorem}\label{thm:approx-dp-algo}
{\sc Approximate-DP} is solvable by a deterministic algorithm using space polynomial in the size of the input program.
\end{theorem}
\begin{proof}
Let $e^\eps = a/b$ for natural numbers $a, b$. As in the algorithm in Section \ref{pspace-pure-dp-algo} we iterate through all pairs of neighboring inputs $(x,x')$, and for each of them compute $$b\delta_{x,x'}(o)=\max(b\Pr[A(x)=o]-a\Pr[A(x')=o],0),$$ using the algorithm from Theorem \ref{hitting_prob}. Then we add this value to the sum $$\sum\limits_{o \in \{0, 1\}^l \cup \{\bot\}}{b\delta_{x,x'}(o)},$$ until we have iterated over all possible inputs, or until the partial sum is greater than $b\delta$. In the former case we terminate with the output ``not DP'', otherwise we do not terminate until checking the last output, and output ``DP''.

Again, the necessary arithmetic computations (maximum, addition, subtraction, and multiplication) on exponentially long rational numbers can be done in polynomial space. 
\end{proof}

\subsection{PSPACE-Hardness}\label{sec:dp-ver-pspace-hard}
To show $\PSPACE$-hardness of checking whether a BPWhile program is differentially private, we reduce from the problem of checking almost sure termination. All of our hardness results have a similar structure: for a program $C$ we construct another program $C'$ that is differentially private (with some parameters) if and only if program $C$ terminates almost surely. We show such reductions for the problems of {\sc Pure-DP}, {\sc Approximate-DP}, and for {\sc Distinguish ($\eps$,$\delta$)}-DP that hold even when the parameters $e^\eps$ and $\delta$ are fixed.

\begin{lemma}\label{lossy_imply_pure_DP}
For a fixed rational $e^\eps > 1$, the problem of checking almost sure termination for BPWhile Boolean programs is poly-time Karp-reducible to the problem of checking  $(\eps,0)$-differential privacy for those programs.
\end{lemma}
\begin{proof}
Let $C(x)$ be a BPWhile program for which we want to check almost sure termination. We construct a new program $C'$ that will receive an input $x$ and one additional bit of input $b$, and runs $C$ as a subroutine. The BPWhile language doesn't support procedure calls, but we can encode the same behavior using the following code representing a template for the code of the program $C'$.

In this reduction we consider two inputs to a program $C'$ as neighboring if they disagree only in one bit.

Here is the template code for $C'$:
$$
    {\tt input}(x,b); {\tt if}\: b == 1\ {\tt then}\ C(x)\ {\tt else}\ {\tt skip};{\tt return}(1)   
$$
Notice that the return statement is executed only if $C(x)$ halts or $b==0$.

As we add constant number of extra lines to the original program $C$, it takes linear time to construct $C'$. Hence the reduction takes linear time.

To show correctness of the reduction we need to check that it maps yes-instances of the almost sure termination problem to yes-instances of {\sc Decide ($\eps$,$0$)}-DP problem, and no-instances to no-instances. If a program $C$ almost surely terminates on all inputs, then for all possible values of bit $b$ we get that $C'(x,b)$ outputs 1 with probability 1. Hence this program is $(\eps,0)$-differentially private for every $\eps \ge 0$.

If the program $C$ is not almost surely terminating, then there exists an input $x$ such that program $C(x)$ fails to terminate with some probability $\rho > 0$. Hence, we get that: $$\Pr[C'(x,1)\text{ doesn't halt}] = \rho > 0.$$
On the other hand, on the neighboring input $(x,0)$ we get
 $$\Pr[C'(x,0)\text{ doesn't halt}] = 0.$$
 Therefore $C'$ is not $\eps$-DP for any $\eps$. 
\end{proof}

\begin{lemma}\label{lossy_imply_approx_DP}
For any fixed rational $e^\eps$ and dyadic $\delta \in (0,1)$, the problem of checking almost sure termination for BPWhile Boolean programs is poly-time Karp-reducible to the problem of checking  $(\eps,\delta)$-differential privacy for those programs.
\end{lemma}

\begin{proof}
As in the proof of Lemma \ref{lossy_imply_pure_DP}, let $C(x)$ be a BPWhile program for which we want to check almost sure termination. We construct a new program $C'$ that will receive an input $x$ and one additional bit of input $b$, that runs $C$ as a subroutine. We denote by $\tt delta\_rand$ a subroutine that outputs $1$ with probability $1-\delta$, and outputs 0 with probability $\delta$. For any dyadic rational constant $\delta = a/2^m$, this can be constructed using $m$ calls to the the random operator. Note that the length of the program computing this subroutine is a constant independent of the length of the input program $C$. The following is a template for the code of $C'$:

\begin{tabular}{ll}
1.& ${\tt input}(x,b);$\\
2.& ${\tt if}\ b == 1\ {\tt then}$\\
3.& \quad $C(x);$\\   
4.& \quad $r = {\tt delta\_rand()};$\\ 
5.& \quad ${\tt if}\ r == 0\ {\tt then}$\\
6.& \quad \quad ${\tt while}\ {\tt true}\ {\tt then}$\\ 
7.& \quad \quad \quad ${\tt skip};$\\ 
8.& \quad ${\tt else}\ {\tt skip};$\\ 
9.& ${\tt else}\ {\tt skip};$\\ 
10.& ${\tt return}(1)$\\
\end{tabular}

Notice that the while-loop in line 6-7 is potentially infinite.
As $\tt delta\_rand$ can be computed by a program of constant size, this reduction takes linear time as in the analysis of Lemma \ref{lossy_imply_pure_DP}. Now to analyze the correctness of the reduction, first assume that $C$ almost surely terminates. Then $C'(x,b)$ either outputs 1 with probability 1, or it outputs $1$ with probability $1-\delta$ and doesn't halt with probability $\delta$. For every pair of input $(x, b), (x', b')$, the statistical distance between the possible distributions on outputs is at most $\delta$.
Hence $C'(x,b)$ is $(0,\delta)$-DP, hence $(\eps,\delta)$-DP.

If $C(x)$ doesn't almost surely terminate, then there exists some $\alpha > 0$ such that on some input $x$ program $C(x)$ enters an infinite loop with probability $\alpha$. Hence overall we get that $C'(x,1)$ enters an infinite loop with probability at least $\delta(1-\alpha) + \alpha > \delta$, but $C'(x,0)$ terminates and outputs 1 with probability $1$. Hence we get that 
$$\Pr[C'(x,1)\text{ doesn't halt}] > \delta = e^{\eps}\Pr[C'(x,0)\text{ doesn't halt}] + \delta,$$
and therefore $C'(x,b)$ is not $(\eps,\delta)$-DP.
\end{proof}

Combining the algorithms from Theorem~\ref{thm:pure-dp-alg} and Theorem~\ref{thm:approx-dp-algo} with the fact that the problems {\sc Pure-DP} and {\sc Approximate-DP} are $\PSPACE$-hard even for fixed values of the privacy parameters, we conclude that corresponding verification problems are $\PSPACE$-complete.

\begin{corollary}
For any rational $\eps$ and dyadic $\delta \in (0,1)$ the problems of checking whether a BPWhile  program is $\eps$-DP or whether a BPWhile program is $(\eps,\delta)$-DP are both $\PSPACE$-complete.
\end{corollary}

\section{Hardness of Approximation of Privacy}\label{sec:approx-dp-param-pspace-hard}
In this section, we show a strong sense in which the privacy parameters of a BPWhile program are hard even to approximate. We do this by showing that for any constant parameters $\eps, \delta$, it is $\PSPACE$-hard even to distinguish between the case where a program is $(0, 0)$-DP or whether it fails to be $(\eps, \delta)$-DP. This, for example, implies that the privacy parameters of a program are hard to approximate up to an additive $(\eps/2, \delta/2)$.

\begin{lemma}\label{lossy_imply_approx_hardness}
For any rational constants $\eps, \delta \in (0,1)$ 
the problem of checking almost sure termination for BPWhile programs is Karp-reducible to the promise problem of determining whether BPWhile program is $(0,0)$-differentially private or it is not $(\eps,\delta)$-differentially private. 
\end{lemma}
\begin{proof}
Our reduction consists of two parts: 
\begin{enumerate}
    \item Given a BPWhile program $C$, we construct a new BPWhile program $C'$ such that if $C$ almost surely terminates, then $C'$ almost surely terminates too. Meanwhile, if $C$ doesn't almost surely terminate, then $C'$ terminates with probability at most $\frac{1}{2}$.
    \item As in the reductions in the proofs of Lemma \ref{lossy_imply_pure_DP} and Lemma \ref{lossy_imply_approx_DP} we construct a program $C''$ that calls $C'$ with the property that $C''$ is $(0, 0)$-DP if $C'$ is almost sure terminating, and $C''$ is not $(\eps,\delta)$-DP if $C'$ halts with probability at most $\frac{1}{2}$.
\end{enumerate}

For the first step, we use the following claim that we prove in Appendix~\ref{app:adp-proofs}.

\begin{claim}\label{loss_amplification}
If $C$ is a BPWhile program, then we can construct in polynomial-time a new program $C'$ that almost surely terminates if $C$ almost surely terminates, and terminates with probability at most $1/2$ if $C$ is not almost surely terminating.
\end{claim}
As in the reductions in Lemma \ref{lossy_imply_pure_DP} and Lemma \ref{lossy_imply_approx_DP}, we now construct a new program $C''$ that  receives an input $x$ and one additional bit of input $b$, and runs $C'$ as a subprogram. We repeat the execution of $C'(x)$ a total of $m$ times, where $e^{\eps} \cdot 2^{-m} + \delta < 1$. Note that $m$ depends only on the privacy parameters $\eps, \delta$ and not on the program $C'$. Now we construct the following program $C''$.

\begin{tabular}{ll}
1.& ${\tt input}(x,b);$\\
2.& ${\tt if}\ b == 1\ {\tt then}$\\
3.& \quad $C'(x);\qquad   \qquad \text{ \# \ run \  $C'(x)$ \ $m$ \ times}$ \\
4.& \quad $...$\\ 
5.& \quad $C'(x);$\\ 
6.& ${\tt else}\ {\tt skip};$\\ 
7.& ${\tt return}(1)$\\
\end{tabular} \\ 

The time complexity of constructing $C''(x,b)$ is linear in the size of $C'$ as $m$ is a constant. To show  correctness,  assume that $C$, and hence $C'$ terminates almost surely. Then $C''(x,b)$ outputs 1 with probability 1 on all inputs. Hence $C''(x,b)$ is $(0, 0)$-DP.

If $C(x)$ does not terminate almost surely then there exist an input $x$ and some $\alpha > 1/2$ such that $C'(x)$ fails to halt with probability $\alpha$. As we chose the number of repetitions $m$ in such way that $$e^{\eps}\cdot \Pr[m \text{ sequential runs of } C'(x) \text{ halt}] + \delta < 1,$$ we get that $C''(x,b)$ is not differentially private on neighboring inputs $(x,0)$ and $(x,1)$, since $$e^{\eps}\cdot \Pr[C''(x,1) \text{ halts}] + \delta < 1 = \Pr[C''(x,0) \text{ halts}].$$

Therefore, if the original program $C(x)$ is not almost surely terminating, we transformed it via the intermediate program $C'(x)$ to a program $C''(x,b)$ that is not $(\eps, \delta)$-differentially private.
\end{proof}

\begin{corollary} 
For any rational constants $e^\eps,\delta$ the problem {\sc Distinguish $(\eps,\delta)$-DP} is $\PSPACE$-hard.
\end{corollary}

\section{Other Definitions of Differential Privacy} \label{sec:other-defs}
Pure and approximate differential privacy degrade smoothly under composition: the overall privacy guarantee of a sequence of DP algorithms remains DP. However, in the worst case it is $\#\P$-hard to compute the best possible parameters achievable by a composition of approximate differentially private algorithms \cite{MurtaghV16}. Other variants of differential privacy, such as R\'enyi \cite{Mironov17}, concentrated \cite{BunSteinke16,DworkRothblum16}, and truncated concentrated differential privacy \cite{BDRS18}, were introduced, in part, to address this problem. All of these notions lead to efficiently computable optimal composition bounds.

We show $\PSPACE$-completeness for each of the problems of verifying (up to a precision parameter given as input) whether a BPWhile program is R\'enyi differentially private, concentrated differentially private, or truncated concentrated differentially private.

\subsection{R\'enyi Differential Privacy}\label{sec:pspace-rdp-algo}
\begin{definition}\label{def-Renyi-div}
Let $P = (p_1, \dots, p_n)$
 and $Q = (q_1, \dots, q_n)$ be probability distributions over $1, \dots, n$. For $\alpha > 1$, the \textit{R\'enyi divergence} of $P$ from $Q$ is
$$D_{\alpha }(P\|Q)={\frac {1}{\alpha -1}}\log {\Bigg (}\sum _{i=1}^{n}{\frac {p_{i}^{\alpha }}{q_{i}^{\alpha -1}}}{\Bigg )}.$$
\end{definition}

\begin{definition}\cite{Mironov17}\label{def-CDP}
A program $C$ is $(\alpha, \rho\alpha)$-R\'enyi-DP if for all neighboring inputs $x,x'$,
$$D_{\alpha }(C(x)\|C(x'))\leq \rho\alpha.$$
\end{definition}

We can check whether a BPWhile program $C$ is $(\alpha, \rho\alpha)$-R\'enyi-DP using an algorithm similar to the one for checking $(\eps,\delta)$-DP from Section \ref{pspace-approx-dp-algo}. A technical issue that arises here is that when computing R\'enyi divergences, we need to exponentiate possibly exponentially long numbers to exponentially large degrees $\alpha$ and $\alpha - 1$. We do not have the space to perform such computations exactly, so instead we consider a ``gappped promise'' version of the problem which takes an additional precision parameter $\eta$ as input, and distinguishes between the case where the program is $(\rho, \rho\alpha)$-RDP and the case where it fails to be $(\rho, \rho\alpha + 2^{-\eta})$-RDP. The inclusion of this precision parameter allows us to approximately compute R\'enyi divergences via additions of logarithms of exponentially long numbers to at most exponential precision.
 
\begin{definition}
In the {\sc Gap-R\'enyi-DP} problem, an instance $(C, \alpha, \rho, \eta)$ consists of a BPWhile program $C$, two dyadic rational numbers $\alpha$ and $\rho$, and a binary integer parameter $\eta$. The problem is to distinguish between the following two cases:
\begin{enumerate}
  \item Yes instances: for all neighboring inputs $x,x'$ we have 
$D_{\alpha }(C(x)\|C(x'))\leq \rho\alpha,$ 
  \item No instances: there exists a pair of neighboring inputs $x,x'$ such that 
$D_{\alpha }(C(x)\|C(x')) \geq \rho\alpha + \frac{1}{2^{\eta}}.$
\end{enumerate}

\end{definition}

\begin{theorem}\label{pspace-Renyi-dp-algo}
The gap problem {\sc Gap-R\'enyi-DP} is solvable by a deterministic algorithm using space polynomial in the size of the instance.
\end{theorem}
\begin{proof}
Consider an instance $(C, \alpha, \rho, \eta)$ of the problem {\sc Gap-R\'enyi-DP}. Following the definition of the problem, we iterate through all pairs of neighboring inputs $(x,x')$, and check that R\'enyi divergence is smaller than $\rho \alpha$. If the length of an instance is at most $n$, then, as the length of the program $C$ is bounded by the length of the instance, we have at most $2^n$ possible output values. Denote the set of output values (including the non-termination outcome $\bot$) as $O$. We need to check the following condition:

$$\log{\sum_{o \in O}{\frac{\Pr[C(x)=o]^\alpha}{\Pr[C(x')=o]^{\alpha-1}}}} \leq \rho \alpha (\alpha-1),$$
which we can rewrite as
$$\sum_{o \in O}{\frac{\Pr[C(x)=o]^\alpha}{\Pr[C(x')=o]^{\alpha-1}}} \leq 2^{\rho \alpha (\alpha-1)}.$$

First of all, we observe that for any pair of neighboring inputs $x,x'$, if for some outcome $o$ we have $\Pr[C(x)=o] > 0$ but $\Pr[C(x')=o]=0$, then we automatically have  a no-instance of the problem. This is because for every possible $\rho$ and $\alpha$ we would get
$$\frac{\Pr[C(x)=o]}{\Pr[C(x')=o]} > 2^{\rho \alpha (\alpha - 1)},$$
and so such $C$ is not R\'enyi differentially private for these parameters. So for each pair of neighboring inputs and for each potential outcome, we first check whether at least one of the probabilities is equal to zero, and output ``no-instance'' if the second probability is non-zero. As all probabilities are finite and represented by numerators and denominators of at most exponential length, this can be performed in polynomial space.

Since $\alpha$ and $\rho$ are given as part of the input, the lengths of $\alpha$ and $\rho$ are at most $n$. Hence $\alpha,\rho \leq 2^n$. Therefore, $|\rho \alpha(\alpha-1)| \leq 2^{3n}$ and we can compute $2^{\rho \alpha(\alpha-1)}$ in polynomial space.\footnote{We can print `1' followed by $2^{\rho \alpha(\alpha-1)}$ zeros by using a counter up to $\rho \alpha(1-\alpha)$ to output the correct number of zeros.}

Taking base-2 logarithms of $\Pr[C(x)=o]$ and $\Pr[C(x)=o]$, our goal is to (approximately) determine whether
$$\sum_{o \in O}{2^{\log{\Pr[C(x)=o]}\cdot \alpha - \log{\Pr[C(x')=o]}\cdot(\alpha-1)}} \leq 2^{\rho \alpha (\alpha-1)}.$$

As this is a comparison between a sum of at most $2^n$ numbers and a number of length $2^{3n}$, it suffices to compute $2^{3n}+\frac{\alpha-1}{2^\eta}+n$ bits
of the quantity  $$2^{\log{\Pr[C(x)=o]}\cdot \alpha - \log{\Pr[C(x')=o]}\cdot(\alpha-1)}$$
for each $o \in O$ to determine which of the two cases we are in:
$$1) \sum_{o \in O}{2^{\log{\Pr[C(x)=o]}\cdot \alpha - \log{\Pr[C(x')=o]}\cdot(\alpha-1)}} \leq 2^{\rho \alpha (\alpha-1)}, \text{ or}$$
$$2) \sum_{o \in O}{2^{\log{\Pr[C(x)=o]}\cdot \alpha - \log{\Pr[C(x')=o]}\cdot(\alpha-1)}} \geq 2^{\rho \alpha (\alpha-1) + \frac{\alpha-1}{2^{\eta}}}.$$

As in our previous $\PSPACE$-algorithms, we cannot explicitly store the value of each partial sum in the memory, as each has exponential length. So below, when we say that we ``compute'' an exponentially long number, we mean that we provide a polynomial space procedure that computes every bit of the number if its index is at most $2^{q(n)}$, where $q(n)$ is a fixed polynomial.

Again as in the algorithm for {\sc Decide ($\eps$,$\delta$)} in Section~\ref{pspace-approx-dp-algo}, our goal is to compute a sum of $2^n$ numbers. But now each of this numbers have more complicated form $2^{\log{\frac{\Pr_a}{\Pr_b}} \cdot \alpha + \log{\Pr_b}}$, where $\Pr_a$ and $\Pr_b$ are exponentially long numbers, and $\alpha$ is a dyadic rational number of length at most $n$. For each element of the sum we need only to compute the $2^{3n}+\frac{\alpha-1}{2^\eta}+n$  most significant bits to guarantee that we underestimate each element of the sum by at most $2^{-(\alpha-1)/{2^{\eta}} - n}$. This yields an overall underestimate of the sum of all $2^n$ elements is at most $2^{-(\alpha-1)/{2^{\eta}}}$. Therefore, we underestimate the logarithm of this sum by at most $-(\alpha-1)/{2^{\eta}}$. Hence we always distinguish case 1 from case 2, by performing a comparison\footnote{Exponentially long numbers can be compared in polynomial space by finding the most significant bit on which they differ.} to determine whether $2^{\log{\frac{\Pr_a}{\Pr_b}} \cdot \alpha + \log{\Pr_b}}$ is greater than $2^{\rho \alpha (\alpha-1) + \frac{\alpha-1}{2^{\eta}}}$ or smaller than $ 2^{\rho \alpha (\alpha-1)}$.

All that remains is to show that  we can compute (i.e., give implicit access to each bit of) each term of the form $2^{\log{\frac{\Pr_a}{\Pr_b}} \cdot \alpha + \log{\Pr_b}}$ in polynomial space. Every bit of the integer logarithm can be computed using uniform circuits of polylogarithmic depth~\cite{Reif86} in the length of the input integer and the index of the requested bit, so we can compute numbers of the form $\log(\Pr_a/\Pr_b)$ and $\log{\Pr_b}$. Further, using space-efficient algorithms for addition and multiplication of exponentially long numbers, as in Theorem \ref{pspace-pure-dp-algo}, we can compute $\log{\frac{\Pr_a}{\Pr_b}} \cdot \alpha + \log{\Pr_b}$. Therefore, we can implicitly compute $\log{\frac{\Pr_a}{\Pr_b}} \cdot \alpha + \log{\Pr_b}$ with a polynomial space algorithm. Finally, the exponential function has a representation as power series, and such power series can be computed by uniform families of logarithmic-depth circuits \cite[Corollary 2.2]{Reif86}. So we can exponentiate $2$ to a dyadic rational degree using an algorithm that runs in space logarithmic in the length of the exponent. Combining polynomial space computations we obtain a polynomial space algorithm for computing $2^{\log{\frac{\Pr_a}{\Pr_b}} \cdot \alpha + \log{\Pr_b}}$. 
\end{proof}

To show $\PSPACE$-hardness, we use Theorem~\ref{lossy_imply_approx_hardness}, and the following fact to reduce from {\sc Distinguish $(\eps,\delta)$-DP} to {\sc Gap-Renyi-DP}:

\begin{theorem}[\cite{Mironov17}]\label{th:rdp-implies-approx-dp}
If $C$ is an $(\alpha,\rho \alpha)$-RDP program, it also satisfies $(\rho \alpha + \frac{\log{1/\delta}}{\alpha - 1}, \delta)$-differential privacy for any $\delta \in (0,1)$.
\end{theorem}

Combining this with Theorem~\ref{lossy_imply_approx_hardness}, which states that it is $\PSPACE$-hard to determine whether a BPWhile program is $(0, 0)$-differentially private or not $(\eps, \delta)$-differentially private, we obtain:

\begin{theorem}\label{th:RDP-pspace-hard}
{\sc Gap-R\'enyi-DP} is $\PSPACE$-hard.
\end{theorem}
\begin{proof}
Fix two dyadic rational numbers $\eps, \delta \in (0,1)$. Let $\eta$, $\rho$ and $\alpha$ be positive numbers with finite binary representations such that 
$$0 < \rho \alpha + \frac{\log{(1/\delta)}}{\alpha-1} + \frac{1}{2^{\eta}} < \eps.$$ 
To reduce from {\sc Distinguish $(\eps,\delta)$-DP} to the {\sc Gap-R\'enyi-DP} problem, we map an instance $C$ of {\sc Distinguish $(\eps,\delta)$-DP} to the instance $(C,\alpha,\rho, \eta)$ in deterministic linear time.

To show correctness of this reduction, first consider the case where $C$ is a yes-instance of {\sc Distinguish $(\eps,\delta)$-DP}. That means that $C$ is $(0,0)$-DP. Then the distributions on the outputs of $C$ are identical for every pair of neighboring inputs. Hence $C$ is also $(\alpha,\rho \alpha)$-R\'enyi-DP, so $(C,\alpha,\rho, \eta)$ is a yes-instance of {\sc Gap-R\'enyi-DP}. 

Now, consider the case where $C$ is a no-instance of {\sc Distinguish $(\eps,\delta)$-DP}. We need to show that $(C,\alpha,\rho, \eta)$ is a no-instance of {\sc Gap-R\'enyi-DP}.
We do this by contraposition: If $(C, \alpha, \rho, \eta)$ is not a no-instance, then  for all neighboring inputs $x,x'$ it holds that $D_{\alpha }(C(x)\|C(x')) \leq \rho\alpha + \frac{1}{2^\eta}$. 
Then for $\rho'=\rho+\frac{1}{\alpha 2^\eta}$, we have that $C$ is $(\alpha,\rho'\alpha)$-R\'enyi-DP. 
Then by Theorem~\ref{th:rdp-implies-approx-dp}, $C$ is $(\rho' \alpha + \frac{\log{(1/\delta)}}{\alpha-1},\delta)$-DP. But by our choice of the parameters, $\rho' \alpha + \frac{\log{(1/\delta)}}{\alpha-1}= \rho \alpha + \frac{1}{2^\eta} + \frac{\log{(1/\delta)}}{\alpha-1} < \eps$, hence $C$ is $(\eps,\delta)$-DP. This implies that $C$ is not a no-instance of {\sc Distinguish $(\eps,\delta)$-DP}.

Since we showed in Theorem~\ref{lossy_imply_approx_hardness} we showed that {\sc Distinguish $(\eps,\delta)$-DP} is $\PSPACE$-hard, it follows that {\sc Gap-R\'enyi-DP} is $\PSPACE$-hard.
\end{proof}

Combining the results of Theorem~\ref{pspace-Renyi-dp-algo} and Theorem~\ref{th:RDP-pspace-hard} we obtain:
\begin{corollary}
{\sc Gap-R\'enyi-DP} is $\PSPACE$-complete.
\end{corollary}

\subsection{Concentrated Differential Privacy}\label{sec:pspace-cdp-algo}
\begin{definition}\cite{BunSteinke16,DworkRothblum16}
A program $C$ is $\rho$-Concentrated-DP if for every neighboring inputs $x,x'$ and every $\alpha \in (1,+\infty)$
$$D_{\alpha }(C(x)\|C(x'))\leq \rho \alpha.$$
\end{definition}

As in Section \ref{sec:pspace-rdp-algo} we consider a gapped version of the problem for an integer precision parameter $\eta$ provided as input.

\begin{definition}
An instance of the {\sc Gap-Concentrated-DP} problem  $(C, \rho, \eta)$ consists of a BPWhile program $C$, a dyadic rational number $\rho$, and a binary integer precision parameter $\eta$. The goal is to distinguish between the following two cases:
\begin{enumerate}
  \item Yes-instances: for all neighboring inputs $x,x'$ and for all $\alpha \in (1,+\infty)$, we have
$D_{\alpha }(C(x)\|C(x'))\leq \rho\alpha,$ 
  \item No-instances: there exists a pair of neighboring inputs $x,x'$ and $\alpha \in (1,+\infty)$ such that 
$D_{\alpha }(C(x)\|C(x')) \geq \rho\alpha + \frac{1}{2^{\eta}}.$
\end{enumerate}
\end{definition}

In order to verify whether an instance $(C,\rho,\eta)$ is a yes-instance of {\sc Gap-Concentrated-DP} we should verify whether the inequality $D_\alpha(C(x) \| C(x')) \leq \rho \alpha$ holds not only for all pairs $(x,x')$, but also for all $\alpha > 1$. This is equivalent to verifying that $C$ is $(\alpha,\rho\alpha)$-R\'enyi-DP for every $\alpha > 1$. But as there is an unbounded continuum of possible $\alpha$ to consider, we do not immediately obtain an algorithm by attempting to exhaustively check them. The following lemma shows that to solve the gapped version of the problem, we need only to consider finitely many $\alpha$ within a finite range.

\begin{lemma}\label{lemma:conc-DP-bounded-alpha}
Suppose $(C, \rho, \eta)$ is a no-instance of the {\sc Gap-Concentrated-DP} problem.  Then there exists a polynomial $p(n)$, neighboring inputs $x, x'$, and $\alpha \in (1, 1 + 2^{p(n)} / \rho)$ an integer multiple of $2^{-\eta-1} / \rho$ such that $D_\alpha(C(x) \| C(x')) \geq \rho \alpha + \frac{1}{2^{\eta+1}}$.
\end{lemma}
\begin{proof}
The problem we are interested in is as follows. Given implicit descriptions of two finite probability distributions $d_1(i)$ and $d_2(i)$, where each probability is discretized to $1/2^{2^{p(n)}}$, and a rational parameter $\rho$, determine whether
$$\frac{1}{\alpha} D_{\alpha}(d_1 \| d_2) = \frac{1}{\alpha(\alpha-1)} \log \sum_{i} (d_1(i))^{\alpha} (d_2(i))^{1-\alpha} \le \rho,$$
for all dyadic rational $\alpha \in (1, \infty)$ with precision $1/2^{p(|x|)}$, or whether for at least one dyadic rational $\alpha \in (1, \infty)$ with  precision $1/2^{p(|x|)}$ it holds that $$\frac{1}{(\alpha-1)} \log \sum_{i} (d_1(i))^{\alpha} (d_2(i))^{1-\alpha} \geq \rho\alpha + \frac{1}{2^{\eta}}.$$
Let $m=2^{2^{p(n)}}$, so each probability in $d_1(i)$ and $d_2(i)$ are discretized to $1/m$. We claim that it suffices to check this condition for all $\alpha < 1 + \log{m} /\rho$. Either 
\begin{enumerate}
     \item There exists $i$ in the probability space such that $d_1(i) > 0$ and $d_2(i) = 0$, in which case $D_{\alpha}(d_1 \| d_2)$ is infinite for every $\alpha > 1$; or
     \item For every outcome $i$, the value $m$ is an upper bound on the ratio $d_1(i) / d_2(i)$. In this case, the quantity $\frac{1}{\alpha} D_{\alpha}(d_1 \| d_2)$ we are interested in is at most 
     \begin{align*}
     \frac{1}{\alpha(\alpha-1)} \log{\sum_{i} m^{\alpha} \cdot d_2(i)} \le&
     \frac{1}{\alpha(\alpha-1)} \log (m^{\alpha}) \le \frac{\log{m}}{\alpha - 1}&,
     \end{align*}
 which is at most $\rho$ for $\alpha \ge 1 + \frac{\log{m}}{\rho}$.
\end{enumerate}

In both cases we need to check values of $\alpha$ within the interval $(1,1 + \log{m/\rho})$.

That still leaves us with the infinite number of values of $\alpha$ we need to check. To finish the proof of the lemma we show that it is enough to consider values of $\alpha$ discretized to $\frac{2^{-\eta}}{\rho}$:

\begin{claim}
Fix distributions $P$ and $Q$ and $\rho > 0$. If $D_{\alpha}(P \| Q) \le \rho \alpha + 2^{-\eta}$ for every $\alpha > 1$ that is discretized to an integer multiple of $2^{-\eta} / \rho$, then $D_{\alpha}(P \| Q) < \rho \alpha + 2^{-\eta + 1}$ for every $\alpha > 1$.
\end{claim}
\begin{proof}
It suffices to show that if $0 < \alpha_1 < \alpha_2 = \alpha_1 + 2^{-\eta}/\rho$ are such that $D_{\alpha_1}(P \| Q) \leq \rho \alpha_1 + 2^{-\eta}$ and $D_{\alpha_2}(P \| Q) \leq \rho \alpha_2 + 2^{-\eta}$, then for every $\alpha'$ with $\alpha_1 < \alpha' < \alpha_2$ we have $D_{\alpha'}(P \| Q) < \rho \alpha' +2^{-\eta+1}$. 
As the R\'enyi divergence between two distributions increases monotonically as a function of $\alpha$, we have
\begin{align*}
D_{\alpha'}(P \| Q) \leq D_{\alpha_2}(P \| Q) &\leq \rho \alpha_2 + 2^{-\eta} \\
= \rho (\alpha' + (\alpha_2 - \alpha')) + 2^{-\eta}
&= \rho \alpha' + \rho(\alpha_2 - \alpha') + 2^{-\eta}.
\end{align*}
From the facts that $\alpha_2-\alpha_1=2^{-\eta}/\rho$ and $\alpha_1 < \alpha' < \alpha_2$ we get $(\alpha_2-\alpha') < 2^{-\eta}/\rho$. Hence $D_{\alpha'}(P \| Q) < \rho \alpha' + 2^{-\eta+1}$.
\end{proof}

This completes the proof of the lemma, as we showed that we can consider values of $\alpha$ discretized to $2^{-\eta}/\rho$ in a bounded range.
\end{proof}

Combining this lemma with the algorithm from Theorem \ref{pspace-Renyi-dp-algo} we get the following theorem:

\begin{theorem}\label{pspace-concentrated-dp-algo}
The gap problem {\sc Gap-Concentrated-DP} is solvable by a deterministic algorithm using space polynomial in the size of the instance.
\end{theorem}

Similarly to the proof of Theorem~\ref{th:RDP-pspace-hard} we can use the fact that concentrated DP implies approximate DP to give a reduction from {\sc Distinguish $(\eps,\delta)$-DP} to {\sc Gap-Concentrated-DP}.
\begin{theorem}[\cite{BunSteinke16}]\label{th:cdp-implies-approx-dp}
If $C$ is an $\rho$-CDP program, it also satisfies $(\rho + 2 \sqrt{\rho\log{(1/\delta)}}, \delta)$-differential privacy for any $\delta \in (0,1)$.
\end{theorem}

\begin{theorem}\label{th:CDP-pspace-hard}
{\sc Gap-Concentrated-DP} is $\PSPACE$-hard.
\end{theorem}

Combining the results of Theorem~\ref{pspace-concentrated-dp-algo} and Theorem~\ref{th:CDP-pspace-hard} we get the corollary:
\begin{corollary}
{\sc Gap-Concentrated-DP} is $\PSPACE$-complete.
\end{corollary}

\subsection{Truncated Concentrated Differential Privacy}\label{sec:pspace-tcdp-algo}

\begin{definition}\cite{BDRS18}
A program $C$ is $\omega$-Truncated $\rho$-Concentrated-DP if for every neighboring inputs $x,x'$ and every $\alpha \in (1,\omega)$
$$D_{\alpha }(C(x)\|C(x'))\leq \rho \alpha.$$
\end{definition}

Again, we introduce a promise version of the problem. This allows us to consider quantities of fixed precision parameterized by $\eta$: 
\begin{definition}
In the {\sc Gap-Truncated-Concentrated-DP} problem, an instance $(C, \rho, \omega, \eta)$ consists of a BPWhile program $C$, two dyadic rational numbers $\rho$ and $\omega$, and a binary integer precision parameter $\eta$. The goal is to distinguish between the following two cases:
\begin{enumerate}
  \item Yes-instances: for all neighboring inputs $x,x'$ and for all $\alpha \in (1,\omega)$,
$$D_{\alpha }(C(x)\|C(x'))\leq \rho\alpha,$$
  \item No-instances: there exists a pair of neighboring inputs $x,x'$ and $\alpha \in (1,\omega)$ such that 
$$D_{\alpha }(C(x)\|C(x')) \geq \rho\alpha + \frac{1}{2^{\eta}}.$$
\end{enumerate}

\end{definition}

As we need to verify a bounded range of values of the parameter $\alpha$ the verification procedure is analogous to the algorithm for verifying concentrated differential privacy. Therefore, we get the following theorem:

\begin{theorem} \label{pspace-trunc-concentrated-algo}
The promise problem {\sc Gap-Truncated-Concentrated-DP} is solvable by a deterministic algorithm using space polynomial in the size of the instance.
\end{theorem}
Similarly to the proof of Theorem~\ref{th:RDP-pspace-hard} and Theorem~\ref{th:CDP-pspace-hard} we use an existing result connecting the parameters of approximate and truncated concentrated differential privacy to show reduction from {\sc Distinguish $(\eps,\delta)$-DP} to {\sc Gap-Truncated-Concentrated-DP}.

\begin{theorem}[\cite{BDRS18}]\label{th:tcdp-implies-approx-dp}
If $C$ is an $(\rho,\omega)$-TCDP mechanism, then it also satisfies $(\rho + 2 \sqrt{\rho\log{(1/\delta)}}, \delta)$-differential privacy for any $\delta \in (0,1)$ that satisfies $\log{(1/\delta)} \leq (\omega-1)^2\rho$.
\end{theorem}

The proof of hardness is then similar to the hardness result for {\sc Gap-R\'enyi-DP}.

\begin{theorem}\label{th:TCDP-pspace-hard}
{\sc Gap-Truncated-Concentrated-DP} is $\PSPACE$-hard.
\end{theorem}

Combining the results of Theorem~\ref{pspace-trunc-concentrated-algo} and Theorem~\ref{th:TCDP-pspace-hard} we get the corollary:
\begin{corollary}
{\sc Gap-Truncated-Concentrated-DP} is $\PSPACE$-complete.
\end{corollary}

\section{Termination-sensitive vs termination-insensitive differential privacy}
In the definition of differential privacy we have considered in this paper, Definition~\ref{def:differential-privacy},  we considered sets of outcomes over $\{0,1\}^l \cup \{\bot\}$. This definition corresponds to  a ``termination-senstive'' differential privacy model where the adversary can observe the program's termination behavior. As in information flow control, one can also study a ``termination-insensitive'' model where we require Condition~\ref{eqn:dp-ineq}, in Definition~\ref{def:differential-privacy}, to hold only for sets of outcomes  $O \subseteq \{0, 1\}^\ell$ and where the probabilities are conditioned on the program $C$ terminating. It is easy to see that,  up to a factor of $2$, termination-sensitive (pure) $\eps$-differential privacy implies termination-insensitive (pure) $\eps$-differential privacy. Indeed, for  all neighbors $x, x'$ and $O \subseteq \{0, 1\}^\ell$ we have:
\begin{align*}
    \Pr[C(x) \in O \mid C(x) \ne \bot] &= \frac{\Pr[C(x) \in O]}{\Pr[C(x) \ne \bot]} \\
    \le \frac{e^\eps \Pr[C(x') \in O]}{e^{-\eps}\Pr[C(x') \ne \bot]}
    &= e^{2\eps} \Pr[C(x') \in O | C(x') \ne \bot].
\end{align*}
This argument does not work for $(\eps, \delta)$-differential privacy. Indeed, a program that given a boolean input $b$ returns $ b$ with probability $\delta$ and $\bot$ with probability $1-\delta$ is termination-sensitive $(0, \delta)$-differentially private, but is not termination-insensitive $(\eps, \delta)$-differentially private for any $\eps < \infty$ or $\delta < 1$.

Meanwhile, even termination-insensitive pure differential privacy does not imply \\ 
termination-sensitive differential privacy. For example, a program that on input $0$ returns $0$ with probability $1$, and on input $1$ returns  $0$ with probability $0.01$, and $\bot$ with probability $0.99$ is termination-insensitive 0-differentially private, but is not termination-sensitive $(\eps, \delta)$-differentially private for any $\delta < 0.99$.

In this work, we focused on the termination-sensitive model  because it is the most natural from a probabilistic perspective and because it is less susceptible to timing side-channel attack such as the one illustrated in the latter example above. Nevertheless, all of our results can be adapted to work for the termination-insensitive model as well. Our algorithmic results can be adapted by explicitly normalizing the computed probabilities by the probability of termination, which can be computed in $\PSPACE$. Meanwhile, our lower bounds hold by replacing all steps where we explicitly enter an infinite loop with steps where we output a special failure symbol. This modification should be made to both the specification and reduction we consider for almost-sure termination, as well as for our reductions from almost-sure termination to privacy verification problems.

\section{Conclusions and Future work}
In this paper we have shown that the problem of deciding a probabilistic boolean program to be differentially private, for several notions of differential privacy, is $\PSPACE$-complete. In addition we have shown that also an approximate version of this problem is $\PSPACE$-hard. These results can help identify the limitations of automated verification methods. One direction that our results point to is the use of QBF solvers~\cite{QBFSolvers} for reasoning about differential privacy and almost sure termination for BPWhile programs. But first, to apply a QBF solver to verify whether an input BPWhile program is differentially private, the program should be efficiently converted to a quntatified boolean formula. In current work we only showed reduction in the opposite direction, namely we converted a QBF formula to a BPWhile program in Lemma~\ref{TQBF_to_losslessness}. We leave the discussion of the applicability of QBF solvers for future work. 

Our proofs of the $\PSPACE$-hardness results for Gap-R\'enyi-DP, Gap-Concentrated-DP, and Gap-Truncated-Concentrated-DP uses simple reduction from $\PSPACE$-hardness of Distinguish $(\eps,\delta)$-DP that doesn't rely on any specific properties of BPWhile language. The similar reduction can be used to show, for example, $\NP$- and $\coNP$-hardness of Gap-R\'enyi-DP, Gap-Concentrated-DP, and Gap-Truncated-Concentrated-DP for loops-free boolean language from~\cite{GNP20}.

 Results that we discuss show that a problem of checking various properties of probabilistic boolean programs with while loops require exactly polynomial space. A problem that generalizes all these results is the problem of checking whether a probability distributions on the outputs of the program satisfy a property expressed by a polynomial-space computation. The $\PSPACE$-hardness of verification of such generalized property is implied by the almost sure termination. That is so as we have an easily verifiable property of the output distribution, we just need to check whether a sum of the probabilities of the outputs is 1 or not. But can we show that every property of the output distribution can be computed by a polynomial space algorithm, if we are given an implicit access to the distribution through the algorithm that outputs every requested bit of the probability of requested output? 

\section*{Acknowledgments}
We thank Alley Stoughton and the anonymous reviewers for their helpful comments and suggestions.

Mark Bun was supported by NSF grants CCF-1947889 and CNS-2046425. Marco Gaboardi was partially supported by NSF grants CNS-2040215 and CNS-2040249. Ludmila Glinskih was supported by NSF grants CCF-1947889 and CCF-1909612.

\bibliography{main.bib}

\begin{thebibliography}{10}

\bibitem{QBFSolvers}
The quantified boolean formulas satisfiability library.
\newblock http://www.qbflib.org.

\bibitem{albarghouthi2017synthesizing}
Aws Albarghouthi and Justin Hsu.
\newblock Synthesizing coupling proofs of differential privacy.
\newblock {\em Proc. {ACM} Program. Lang.}, 2({POPL}):58:1--58:30, 2018.
\newblock \href {http://dx.doi.org/10.1145/3158146}
  {\path{doi:10.1145/3158146}}.

\bibitem{BalcerV18}
Victor Balcer and Salil~P. Vadhan.
\newblock Differential privacy on finite computers.
\newblock In Anna~R. Karlin, editor, {\em 9th Innovations in Theoretical
  Computer Science Conference, {ITCS} 2018, January 11-14, 2018, Cambridge, MA,
  {USA}}, volume~94 of {\em LIPIcs}, pages 43:1--43:21. Schloss Dagstuhl -
  Leibniz-Zentrum f{\"{u}}r Informatik, 2018.
\newblock \href {http://dx.doi.org/10.4230/LIPIcs.ITCS.2018.43}
  {\path{doi:10.4230/LIPIcs.ITCS.2018.43}}.

\bibitem{BCJSV20}
Gilles Barthe, Rohit Chadha, Vishal Jagannath, A.~Prasad Sistla, and Mahesh
  Viswanathan.
\newblock Deciding differential privacy for programs with finite inputs and
  outputs.
\newblock In {\em {LICS} '20: 35th Annual {ACM/IEEE} Symposium on Logic in
  Computer Science, Saarbr{\"{u}}cken, Germany, July 8-11, 2020}, pages
  141--154. {ACM}, 2020.

\bibitem{BartheGAHKS14}
Gilles Barthe, Marco Gaboardi, Emilio Jes{\'{u}}s~Gallego Arias, Justin Hsu,
  C{\'{e}}sar Kunz, and Pierre{-}Yves Strub.
\newblock Proving differential privacy in hoare logic.
\newblock In {\em {IEEE} 27th Computer Security Foundations Symposium, {CSF}
  2014, Vienna, Austria, 19-22 July, 2014}, pages 411--424. {IEEE} Computer
  Society, 2014.
\newblock \href {http://dx.doi.org/10.1109/CSF.2014.36}
  {\path{doi:10.1109/CSF.2014.36}}.

\bibitem{BartheGAHRS15}
Gilles Barthe, Marco Gaboardi, Emilio Jes{\'{u}}s~Gallego Arias, Justin Hsu,
  Aaron Roth, and Pierre{-}Yves Strub.
\newblock Higher-order approximate relational refinement types for mechanism
  design and differential privacy.
\newblock In {\em {POPL}}, pages 55--68, 2015.

\bibitem{Barthe:2016}
Gilles Barthe, Marco Gaboardi, Benjamin Gr{\'e}goire, Justin Hsu, and
  Pierre-Yves Strub.
\newblock Proving differential privacy via probabilistic couplings.
\newblock In {\em LICS '16}, pages 749--758, New York, NY, USA, 2016. ACM.

\bibitem{barthe2012probabilistic}
Gilles Barthe, Boris K{\"o}pf, Federico Olmedo, and Santiago Zanella~Beguelin.
\newblock Probabilistic relational reasoning for differential privacy.
\newblock {\em ACM SIGPLAN Notices}, 47(1):97--110, 2012.

\bibitem{Bichsel:2018}
Benjamin Bichsel, Timon Gehr, Dana Drachsler-Cohen, Petar Tsankov, and Martin
  Vechev.
\newblock Dp-finder: Finding differential privacy violations by sampling and
  optimization.
\newblock In {\em CCS '18}, pages 508--524, 2018.

\bibitem{BorodinCookPippenger83}
Allan Borodin, Stephen~A. Cook, and Nicholas Pippenger.
\newblock Parallel computation for well-endowed rings and space-bounded
  probabilistic machines.
\newblock {\em Inf. Control.}, 58(1-3):113--136, 1983.
\newblock \href {http://dx.doi.org/10.1016/S0019-9958(83)80060-6}
  {\path{doi:10.1016/S0019-9958(83)80060-6}}.

\bibitem{BDRS18}
Mark Bun, Cynthia Dwork, Guy~N. Rothblum, and Thomas Steinke.
\newblock Composable and versatile privacy via truncated {CDP}.
\newblock In Ilias Diakonikolas, David Kempe, and Monika Henzinger, editors,
  {\em Proceedings of the 50th Annual {ACM} {SIGACT} Symposium on Theory of
  Computing, {STOC} 2018, Los Angeles, CA, USA, June 25-29, 2018}, pages
  74--86. {ACM}, 2018.
\newblock \href {http://dx.doi.org/10.1145/3188745.3188946}
  {\path{doi:10.1145/3188745.3188946}}.

\bibitem{BunSteinke16}
Mark Bun and Thomas Steinke.
\newblock Concentrated differential privacy: Simplifications, extensions, and
  lower bounds.
\newblock In Martin Hirt and Adam~D. Smith, editors, {\em Theory of
  Cryptography - 14th International Conference, {TCC} 2016-B, Beijing, China,
  October 31 - November 3, 2016, Proceedings, Part {I}}, volume 9985 of {\em
  Lecture Notes in Computer Science}, pages 635--658, 2016.
\newblock \href {http://dx.doi.org/10.1007/978-3-662-53641-4\_24}
  {\path{doi:10.1007/978-3-662-53641-4\_24}}.

\bibitem{ChadhaKV14}
Rohit Chadha, Dileep Kini, and Mahesh Viswanathan.
\newblock Quantitative information flow in boolean programs.
\newblock In Mart{\'{\i}}n Abadi and Steve Kremer, editors, {\em Principles of
  Security and Trust - Third International Conference, {POST} 2014, Held as
  Part of the European Joint Conferences on Theory and Practice of Software,
  {ETAPS} 2014, Grenoble, France, April 5-13, 2014, Proceedings}, volume 8414
  of {\em Lecture Notes in Computer Science}, pages 103--119. Springer, 2014.
\newblock \href {http://dx.doi.org/10.1007/978-3-642-54792-8\_6}
  {\path{doi:10.1007/978-3-642-54792-8\_6}}.

\bibitem{ChadhaSV21}
Rohit Chadha, A.~Prasad Sistla, and Mahesh Viswanathan.
\newblock On linear time decidability of differential privacy for programs with
  unbounded inputs.
\newblock In {\em {LICS}}, pages 1--13. {IEEE}, 2021.

\bibitem{chistikov2018bisimilarity}
Dmitry Chistikov, Andrzej~S. Murawski, and David Purser.
\newblock Bisimilarity distances for approximate differential privacy.
\newblock In Shuvendu~K. Lahiri and Chao Wang, editors, {\em Automated
  Technology for Verification and Analysis - 16th International Symposium,
  {ATVA} 2018, Los Angeles, CA, USA, October 7-10, 2018, Proceedings}, volume
  11138 of {\em Lecture Notes in Computer Science}, pages 194--210. Springer,
  2018.
\newblock \href {http://dx.doi.org/10.1007/978-3-030-01090-4_12}
  {\path{doi:10.1007/978-3-030-01090-4_12}}.

\bibitem{chistikov2019asymmetric}
Dmitry Chistikov, Andrzej~S. Murawski, and David Purser.
\newblock Asymmetric distances for approximate differential privacy.
\newblock In Wan Fokkink and Rob van Glabbeek, editors, {\em 30th International
  Conference on Concurrency Theory, {CONCUR} 2019, August 27-30, 2019,
  Amsterdam, the Netherlands}, volume 140 of {\em LIPIcs}, pages 10:1--10:17.
  Schloss Dagstuhl - Leibniz-Zentrum f{\"{u}}r Informatik, 2019.
\newblock \href {http://dx.doi.org/10.4230/LIPIcs.CONCUR.2019.10}
  {\path{doi:10.4230/LIPIcs.CONCUR.2019.10}}.

\bibitem{CourcoubetisYannakakis95}
Costas Courcoubetis and Mihalis Yannakakis.
\newblock The complexity of probabilistic verification.
\newblock {\em J. {ACM}}, 42(4):857--907, 1995.
\newblock \href {http://dx.doi.org/10.1145/210332.210339}
  {\path{doi:10.1145/210332.210339}}.

\bibitem{DingWWZK18}
Zeyu Ding, Yuxin Wang, Guanhong Wang, Danfeng Zhang, and Daniel Kifer.
\newblock Detecting violations of differential privacy.
\newblock In {\em {CCS} 2018}, pages 475--489, 2018.

\bibitem{DingWZK19}
Zeyu Ding, Yuxin Wang, Danfeng Zhang, and Dan Kifer.
\newblock Free gap information from the differentially private sparse vector
  and noisy max mechanisms.
\newblock {\em Proc. {VLDB} Endow.}, 13(3):293--306, 2019.
\newblock \href {http://dx.doi.org/10.14778/3368289.3368295}
  {\path{doi:10.14778/3368289.3368295}}.

\bibitem{DMNS06}
Cynthia Dwork, Frank McSherry, Kobbi Nissim, and Adam~D. Smith.
\newblock Calibrating noise to sensitivity in private data analysis.
\newblock In {\em Theory of Cryptography, Third Theory of Cryptography
  Conference, {TCC} 2006, New York, NY, USA, March 4-7, 2006, Proceedings},
  volume 3876 of {\em Lecture Notes in Computer Science}, pages 265--284.
  Springer, 2006.

\bibitem{DworkRothblum16}
Cynthia Dwork and Guy~N. Rothblum.
\newblock Concentrated differential privacy.
\newblock {\em CoRR}, abs/1603.01887, 2016.
\newblock \href {http://arxiv.org/abs/1603.01887} {\path{arXiv:1603.01887}}.

\bibitem{EtessamiY09}
Kousha Etessami and Mihalis Yannakakis.
\newblock Recursive markov chains, stochastic grammars, and monotone systems of
  nonlinear equations.
\newblock {\em J. {ACM}}, 56(1):1:1--1:66, 2009.
\newblock \href {http://dx.doi.org/10.1145/1462153.1462154}
  {\path{doi:10.1145/1462153.1462154}}.

\bibitem{FarinaCG20}
Gian~Pietro Farina, Stephen Chong, and Marco Gaboardi.
\newblock Coupled relational symbolic execution for differential privacy.
\newblock In {\em {ESOP}}, volume 12648 of {\em Lecture Notes in Computer
  Science}, pages 207--233. Springer, 2021.

\bibitem{fredrikson2014satisfiability}
Matthew Fredrikson and Somesh Jha.
\newblock Satisfiability modulo counting: a new approach for analyzing privacy
  properties.
\newblock In Thomas~A. Henzinger and Dale Miller, editors, {\em Joint Meeting
  of the Twenty-Third {EACSL} Annual Conference on Computer Science Logic
  {(CSL)} and the Twenty-Ninth Annual {ACM/IEEE} Symposium on Logic in Computer
  Science (LICS), {CSL-LICS} '14, Vienna, Austria, July 14 - 18, 2014}, pages
  42:1--42:10. {ACM}, 2014.
\newblock \href {http://dx.doi.org/10.1145/2603088.2603097}
  {\path{doi:10.1145/2603088.2603097}}.

\bibitem{Gaboardi2013}
Marco Gaboardi, Andreas Haeberlen, Justin Hsu, Arjun Narayan, and Benjamin~C.
  Pierce.
\newblock Linear dependent types for differential privacy.
\newblock In {\em POPL}, pages 357--370, 2013.

\bibitem{GNP20}
Marco Gaboardi, Kobbi Nissim, and David Purser.
\newblock The complexity of verifying loop-free programs as differentially
  private.
\newblock In {\em 47th International Colloquium on Automata, Languages, and
  Programming, {ICALP} 2020, July 8-11, 2020, Saarbr{\"{u}}cken, Germany
  (Virtual Conference)}, volume 168 of {\em LIPIcs}, pages 129:1--129:17.
  Schloss Dagstuhl - Leibniz-Zentrum f{\"{u}}r Informatik, 2020.

\bibitem{GazeauMP16}
Ivan Gazeau, Dale Miller, and Catuscia Palamidessi.
\newblock Preserving differential privacy under finite-precision semantics.
\newblock {\em Theor. Comput. Sci.}, 655:92--108, 2016.
\newblock \href {http://dx.doi.org/10.1016/j.tcs.2016.01.015}
  {\path{doi:10.1016/j.tcs.2016.01.015}}.

\bibitem{GhoshRS12}
Arpita Ghosh, Tim Roughgarden, and Mukund Sundararajan.
\newblock Universally utility-maximizing privacy mechanisms.
\newblock {\em {SIAM} J. Comput.}, 41(6):1673--1693, 2012.
\newblock \href {http://dx.doi.org/10.1137/09076828X}
  {\path{doi:10.1137/09076828X}}.

\bibitem{GilbertM18}
Anna~C. Gilbert and Audra McMillan.
\newblock Property testing for differential privacy.
\newblock In {\em 56th Annual Allerton Conference on Communication, Control,
  and Computing, Allerton 2018, Monticello, IL, USA, October 2-5, 2018}, pages
  249--258. {IEEE}, 2018.
\newblock \href {http://dx.doi.org/10.1109/ALLERTON.2018.8636068}
  {\path{doi:10.1109/ALLERTON.2018.8636068}}.

\bibitem{GY13}
Patrice Godefroid and Mihalis Yannakakis.
\newblock Analysis of boolean programs.
\newblock In {\em Tools and Algorithms for the Construction and Analysis of
  Systems - 19th International Conference, {TACAS} 2013, Held as Part of the
  European Joint Conferences on Theory and Practice of Software, {ETAPS} 2013,
  Rome, Italy, March 16-24, 2013. Proceedings}, volume 7795 of {\em Lecture
  Notes in Computer Science}, pages 214--229. Springer, 2013.

\bibitem{HartSharir84}
Sergiu Hart and Micha Sharir.
\newblock Probabilistic temporal logics for finite and bounded models.
\newblock In Richard~A. DeMillo, editor, {\em Proceedings of the 16th Annual
  {ACM} Symposium on Theory of Computing, April 30 - May 2, 1984, Washington,
  DC, {USA}}, pages 1--13. {ACM}, 1984.
\newblock \href {http://dx.doi.org/10.1145/800057.808660}
  {\path{doi:10.1145/800057.808660}}.

\bibitem{HartmanisS76}
Juris Hartmanis and Janos Simon.
\newblock On the structure of feasible computations.
\newblock {\em Adv. Comput.}, 14:1--43, 1976.
\newblock \href {http://dx.doi.org/10.1016/S0065-2458(08)60449-0}
  {\path{doi:10.1016/S0065-2458(08)60449-0}}.

\bibitem{Ilvento20}
Christina Ilvento.
\newblock Implementing the exponential mechanism with base-2 differential
  privacy.
\newblock In Jay Ligatti, Xinming Ou, Jonathan Katz, and Giovanni Vigna,
  editors, {\em {CCS} '20: 2020 {ACM} {SIGSAC} Conference on Computer and
  Communications Security, Virtual Event, USA, November 9-13, 2020}, pages
  717--742. {ACM}, 2020.
\newblock \href {http://dx.doi.org/10.1145/3372297.3417269}
  {\path{doi:10.1145/3372297.3417269}}.

\bibitem{Jung81}
H.~Jung.
\newblock Relationships between probabilistic and deterministic tape
  complexity.
\newblock In Jozef Gruska and Michal Chytil, editors, {\em Mathematical
  Foundations of Computer Science 1981, Strbske Pleso, Czechoslovakia, August
  31 - September 4, 1981, Proceedings}, volume 118 of {\em Lecture Notes in
  Computer Science}, pages 339--346. Springer, 1981.
\newblock \href {http://dx.doi.org/10.1007/3-540-10856-4\_101}
  {\path{doi:10.1007/3-540-10856-4\_101}}.

\bibitem{KaminskiKM19}
Benjamin~Lucien Kaminski, Joost{-}Pieter Katoen, and Christoph Matheja.
\newblock On the hardness of analyzing probabilistic programs.
\newblock {\em Acta Informatica}, 56(3):255--285, 2019.
\newblock \href {http://dx.doi.org/10.1007/s00236-018-0321-1}
  {\path{doi:10.1007/s00236-018-0321-1}}.

\bibitem{KaplanMS20}
Haim Kaplan, Yishay Mansour, and Uri Stemmer.
\newblock The sparse vector technique, revisited.
\newblock In {\em {COLT}}, volume 134 of {\em Proceedings of Machine Learning
  Research}, pages 2747--2776. {PMLR}, 2021.

\bibitem{KiferMRTZ20}
Daniel Kifer, Solomon Messing, Aaron Roth, Abhradeep Thakurta, and Danfeng
  Zhang.
\newblock Guidelines for implementing and auditing differentially private
  systems.
\newblock {\em CoRR}, abs/2002.04049, 2020.
\newblock \href {http://arxiv.org/abs/2002.04049} {\path{arXiv:2002.04049}}.

\bibitem{LehmannShelah82}
Daniel Lehmann and Saharon Shelah.
\newblock Reasoning with time and chance.
\newblock {\em Inf. Control.}, 53(3):165--198, 1982.
\newblock \href {http://dx.doi.org/10.1016/S0019-9958(82)91022-1}
  {\path{doi:10.1016/S0019-9958(82)91022-1}}.

\bibitem{LiuWZ18}
Depeng Liu, Bow{-}Yaw Wang, and Lijun Zhang.
\newblock Model checking differentially private properties.
\newblock In Sukyoung Ryu, editor, {\em Programming Languages and Systems -
  16th Asian Symposium, {APLAS} 2018, Wellington, New Zealand, December 2-6,
  2018, Proceedings}, volume 11275 of {\em Lecture Notes in Computer Science},
  pages 394--414. Springer, 2018.
\newblock \href {http://dx.doi.org/10.1007/978-3-030-02768-1\_21}
  {\path{doi:10.1007/978-3-030-02768-1\_21}}.

\bibitem{Lyu-2017}
Min Lyu, Dong Su, and Ninghui Li.
\newblock Understanding the sparse vector technique for differential privacy.
\newblock {\em Proc. VLDB Endow.}, 10(6):637--648, February 2017.

\bibitem{Mironov12}
Ilya Mironov.
\newblock On significance of the least significant bits for differential
  privacy.
\newblock In Ting Yu, George Danezis, and Virgil~D. Gligor, editors, {\em the
  {ACM} Conference on Computer and Communications Security, CCS'12, Raleigh,
  NC, USA, October 16-18, 2012}, pages 650--661. {ACM}, 2012.
\newblock \href {http://dx.doi.org/10.1145/2382196.2382264}
  {\path{doi:10.1145/2382196.2382264}}.

\bibitem{Mironov17}
Ilya Mironov.
\newblock R{\'{e}}nyi differential privacy.
\newblock In {\em 30th {IEEE} Computer Security Foundations Symposium, {CSF}
  2017, Santa Barbara, CA, USA, August 21-25, 2017}, pages 263--275. {IEEE}
  Computer Society, 2017.
\newblock \href {http://dx.doi.org/10.1109/CSF.2017.11}
  {\path{doi:10.1109/CSF.2017.11}}.

\bibitem{MurtaghV16}
Jack Murtagh and Salil~P. Vadhan.
\newblock The complexity of computing the optimal composition of differential
  privacy.
\newblock In Eyal Kushilevitz and Tal Malkin, editors, {\em Theory of
  Cryptography - 13th International Conference, {TCC} 2016-A, Tel Aviv, Israel,
  January 10-13, 2016, Proceedings, Part {I}}, volume 9562 of {\em Lecture
  Notes in Computer Science}, pages 157--175. Springer, 2016.
\newblock \href {http://dx.doi.org/10.1007/978-3-662-49096-9\_7}
  {\path{doi:10.1007/978-3-662-49096-9\_7}}.

\bibitem{NearDASGWSZSSS19}
Joseph~P. Near, David Darais, Chike Abuah, Tim Stevens, Pranav Gaddamadugu, Lun
  Wang, Neel Somani, Mu~Zhang, Nikhil Sharma, Alex Shan, and Dawn Song.
\newblock Duet: an expressive higher-order language and linear type system for
  statically enforcing differential privacy.
\newblock {\em Proc. {ACM} Program. Lang.}, 3({OOPSLA}):172:1--172:30, 2019.

\bibitem{ReedP10}
Jason Reed and Benjamin~C. Pierce.
\newblock Distance makes the types grow stronger: a calculus for differential
  privacy.
\newblock In {\em {ICFP} 2010}, pages 157--168. {ACM}, 2010.

\bibitem{Reif86}
John~H. Reif.
\newblock Logarithmic depth circuits for algebraic functions.
\newblock {\em {SIAM} J. Comput.}, 15(1):231--242, 1986.
\newblock \href {http://dx.doi.org/10.1137/0215017}
  {\path{doi:10.1137/0215017}}.

\bibitem{Savitch70}
Walter~J. Savitch.
\newblock Relationships between nondeterministic and deterministic tape
  complexities.
\newblock {\em J. Comput. Syst. Sci.}, 4(2):177--192, 1970.
\newblock \href {http://dx.doi.org/10.1016/S0022-0000(70)80006-X}
  {\path{doi:10.1016/S0022-0000(70)80006-X}}.

\bibitem{Sim81}
Janos Simon.
\newblock On tape-bounded probabilistic turing machine acceptors.
\newblock {\em Theor. Comput. Sci.}, 16:75--91, 1981.

\bibitem{tschantz2011formal}
Michael~Carl Tschantz, Dilsun~Kirli Kaynar, and Anupam Datta.
\newblock Formal verification of differential privacy for interactive systems
  (extended abstract).
\newblock In Michael~W. Mislove and Jo{\"{e}}l Ouaknine, editors, {\em
  Twenty-seventh Conference on the Mathematical Foundations of Programming
  Semantics, {MFPS} 2011, Pittsburgh, PA, USA, May 25-28, 2011}, volume 276 of
  {\em Electronic Notes in Theoretical Computer Science}, pages 61--79.
  Elsevier, 2011.
\newblock \href {http://dx.doi.org/10.1016/j.entcs.2011.09.015}
  {\path{doi:10.1016/j.entcs.2011.09.015}}.

\bibitem{Vardi85}
Moshe~Y. Vardi.
\newblock Automatic verification of probabilistic concurrent finite-state
  programs.
\newblock In {\em 26th Annual Symposium on Foundations of Computer Science,
  Portland, Oregon, USA, 21-23 October 1985}, pages 327--338. {IEEE} Computer
  Society, 1985.
\newblock \href {http://dx.doi.org/10.1109/SFCS.1985.12}
  {\path{doi:10.1109/SFCS.1985.12}}.

\bibitem{WangDKZ20}
Yuxin Wang, Zeyu Ding, Daniel Kifer, and Danfeng Zhang.
\newblock Checkdp: An automated and integrated approach for proving
  differential privacy or finding precise counterexamples.
\newblock In {\em Proceedings of the 2020 {ACM} {SIGSAC} Conference on Computer
  and Communications Security}, 2020.
\newblock To appear.

\bibitem{Wegener87}
Ingo Wegener.
\newblock {\em The complexity of Boolean functions}.
\newblock Wiley-Teubner, 1987.

\bibitem{ZhangK17}
Danfeng Zhang and Daniel Kifer.
\newblock Lightdp: towards automating differential privacy proofs.
\newblock In {\em {POPL} 2017}, pages 888--901. {ACM}, 2017.

\bibitem{ZhangRHP020}
Hengchu Zhang, Edo Roth, Andreas Haeberlen, Benjamin~C. Pierce, and Aaron Roth.
\newblock Testing differential privacy with dual interpreters.
\newblock {\em Proc. {ACM} Program. Lang.}, 4({OOPSLA}):165:1--165:26, 2020.

\end{thebibliography}

\appendix 
\section{Computing Hitting Probabilities in Polylogarithmic Space}\label{hitting_algo}
Here we give a brief overview of Simon's algorithm~\cite{Sim81} for computing the hitting probabilities of a Markov chain with $n$ states using $O(\log^6 n)$ space.

The original application in Simon's paper was to  show that unbounded-error probabilistic Turing machines can be simulated in deterministic polynomial space. That is, he showed that one can determine in $\PSPACE$ whether the accept configuration in the configuration graph of a probabilistic TM is reached with probability strictly greater than $1/2$.
 This in turn is accomplished by interpreting the configuration graph as a Markov chain and exactly computing the hitting probability of the accept configuration.

We now describe the algorithm for computing hitting probabilities captured in Lemma~\ref{hitting_prob}. Recall that we are given a Markov chain $M = (V, E, p, p_0)$ with $2^{L}$ states. The Markov chain is represented by its transition matrix (an object of size $2^{O(L)}$), so each entry can be addressed using $O(L)$ space. We assume that $p_0$ is supported on a single start state, that all non-final states are non-recurrent (i.e., upon leaving a non-final state, the probability the Markov chain returns to it is less than $1$) and that for all non-final states, every outgoing transition has probability either $0$ or $1/2$.

Simon first described an algorithm using $O(L^3)$ time in the \emph{random access machine with multiplication} (MRAM) model -- a model of parallel computation with unit-cost multiplication. This implies an $O(L^6)$-space algorithm on a deterministic TM using a generic simulation of time $T(n)$ MRAM algorithms by space $O(T^2(n))$-space deterministic TMs~\cite{HartmanisS76}.

The MRAM algorithm works as follows. Let $P$ denote the transition matrix of the Markov chain $M$. Let $Q$ be the submatrix of $P$ corresponding to the non-final states. For a given final state $f$, let $v_f$ be the column of $P$ corresponding to state $f$, but restricted to the entries corresponding to non-final states. Let $v_s^T$ be the row vector with a $1$ in the entry corresponding to the start state and $0$'s elsewhere. Then letting $Q_\infty = Q + Q^2 + Q^3 + \dots$, we have that the probability of reaching the final state $f$ from the start state $s$ is $a = v_s^T Q_\infty V_f$. The goal now becomes to compute this matrix-vector product.

The key idea is that since $Q$ consists only of non-recurrent states, then $Q_\infty$ is well-defined and $Q_\infty = (I - Q)^{-1} - I$. Matrix inversion (more precisely, computing the numerators and denominators of the resulting entries separately) can be performed on an MRAM in time $O(L^3)$ using a variant of Csanksy's algorithm. This dominates the runtime of the algorithm, which just has to perform the matrix-vector product.

\section{Reduction from TQBF to Almost Sure Termination}\label{TQBF_to_losslessness}
Suppose we have a fully quantified Boolean formula $$\psi = \forall x_1 \in \{0,1\} \exists x_2 \in \{0,1\} \dots \forall x_t \in \{0,1\} \phi(x_1,\dots,x_t)$$ in prenex normal form. We wish to check whether $\psi \in \TQBF$. We create a BPWhile program, the template for which we give below, such that the program terminates almost surely iff $\psi$ is true. As in previous reductions, for the sake of readability we use a few extra constructions that BPWhile doesn't formally support, such as variables that take on constant-size integer values.

\begin{lstlisting}
A:
input(b); # dummy input bit that
# the program ignores
c1 = 0;
x1 = 0;
while x1 <= 1  then
    c2 = 0;
    x2 = 0;
    while x2 <= 1 then
        c3 = 0;
        x3 = 0;
        while x3 <= 1 then
        ...
            while xt <= 1 then
                if phi(x1,...,xt)==1 then
                    ct++;
                xt++;
            if ct == 2 then
                c(t-1)++;
        ...
        if c3 == 2 then
            c2++;
        x2++;
    if c2 >= 1 then
        c1++;
    x1++;
if (c1 < 2) then  #psi is false
    while true then #enter infinite loop
        skip;
return(1)
\end{lstlisting}

The first part of the program uses $t$ nested while loops to evaluate the QBF formula $\psi$. Each loop corresponding to a universal quantifier checks that both assignments to its variable return 1. Meanwhile, each loop corresponding to an existential quantifier checks that at least one of the assignments returns 1.

After evaluating the entire formula, the program enters an infinite loop if it evaluates to false, and otherwise terminates with probability $1$.

Hence this construction produces a BPWhile program that terminates with probability $1$ iff the formula $\psi$ is true. The construction of the program takes time polynomial in the size of $\psi$, so checking almost sure termination is $\PSPACE$-hard.

\section{Additional Proofs from Section~\ref{sec:approx-dp-param-pspace-hard}} \label{app:adp-proofs}

\begin{proof}[Proof of Claim~\ref{loss_amplification}]
Suppose $P$ is not almost surely terminating and that $x$ is an input on which the program is not terminating with some positive probability. Consider the Markov chain corresponding to the execution of $P(x)$. This Markov chain has a reachable, recurrent non-final state. Since a program of size $N$ has at most $2^{\poly(N)}$ states, this recurrent state is reachable within $2^{\poly(N)}$ transitions. Moreover, since each transition has probability either $0, 1/2$, or $1$, the probability of reaching this recurrent state is at least $2^{-2^{\poly(N)}}$.

The program $P'$ will amplify the probability of reaching this recurrent state (i.e., entering this infinite loop) by repeating $P$ many times. Below we describe how to encode this number of repetitions succinctly. We provide a code template where we operate with two vectors of $m+1$ boolean variables $Counter$ and $B$, that we use in this program to represent integers in the range $[0,2^{m+1})$. We compare and increment these variables, and both of these operations can be  encoded as simple procedure of polynomial size in the length of $m$ with boolean variables only. 

\begin{tabular}{ll}
1.& ${\tt input}(x);$\\
2.& ${\tt while\ true\ then}$\\
3.&  \quad $B = 0;$\\
4.& \quad $Counter = 0;$\\
5.& \quad ${\tt while}\ (Counter < 2^m)\ {\tt then}$\\
6.& \quad \quad ${\tt increment(Counter)}$;\\
7.& \quad  \quad $a = {\tt random};$\\ 
8.& \quad \quad  ${\tt if}\ a = 1\ {\tt then}$\\
9.& \quad \quad \quad ${\tt increment}(B);$\\
10.& \quad \quad \quad ${\tt if}\  B < 2^m\ {\tt then}$\\
11.& \quad \quad \quad \quad $ P(x);$\\
12.& \quad \quad \quad {\tt else}\\
13.& \quad \quad \quad \quad ${\tt return}(1)$\\  
\end{tabular}

In short, this program terminates if and only if out of $2^m$ coin tosses in the inner while loop, we get $2^m$ 1's. As the probability of this event is $1/2^{2^m}$, we get that we need approximately $2^{2^m}$ iterations of the external loop to finally get exactly $2^m$ ones in $2^m$ coin tosses. As in each iteration of the external loop we run the program $P$ that with probability at least $1/2^{2^{m}}$ enters an infinite loop, overall we enter this loop with constant probability when it exists. 

We now analyze the guarantee of $P'$ more formally. Let $X$ be the number of rounds in which the outer loop runs, and $q = 2^{-2^m}$ be the probability of getting $B=2^m$ after the inner loop run. Then we get that $X$ is distributed as a geometric random variable
$$\Pr[X = 1] = q, 
\Pr[X = 2] = q(1-q),$$
$$\Pr[X = 3] = q(1-q)^2,\dots$$

Then we can estimate the probability that $P'$ halts as:
\begin{align*}
\Pr[P' &\text{ halts}] \leq \sum_{k=1}^\infty q(1-q)^{k-1} \cdot (1 - q)^k
= q \sum_{k = 1}^\infty (1-q)^{2k-1} \\
&= \frac{q}{1-q} \cdot \sum_{k = 1}^\infty (1-q)^{2k}
= \frac{q}{1-q}\cdot \frac{(1-q)^2}{1 - (1-q)^2} \\
&= \frac{q \cdot (1-q)}{1 - (1-q)^2} = \frac{q(1-q)}{q(2-q)} = \frac{1-q}{2-q}=1/2 - \frac{0.5q}{2-q}.
\end{align*}

Hence, if $0 < q < 1$ we get that $\Pr[P' \text{ halts}] < 1/2$.
\end{proof}

\section{Additional Proofs from Section~\ref{sec:other-defs}}\label{app:other-defs}
\begin{proof}[Proof of Theorem~\ref{pspace-concentrated-dp-algo}]
Let $C$ be a BPWhile program of length $n$, $\rho$ be a dyadic rational number, and $\eta$ be a binary integer precision parameter. Let $p(n)$ be a polynomial such that all probabilities of reaching the final state in the Markov chain of $P$ are discretized to $2^{-2^{p(n)}}$. 

By Lemma \ref{lemma:conc-DP-bounded-alpha}, in order to solve an instance $(C,\rho,\eta)$ {\sc Gap-Concentrated-DP}, it suffices to solve instances $(C, \alpha, \rho,\eta)$ of  {\sc Gap-R\'enyi-DP} for every $\alpha$ in range $(1,1+ \frac{2^{p(n)}}{\rho})$ discretized to precision $2^{-\eta-1}/\rho$. That will also imply that we never have an event where $x,x'$ are neighboring inputs, $\Pr[C(x)=o]=0$ and $\Pr[C(x')=o] > 0$. Note that the largest value of $\alpha$ is $1+2^{p(n)}/\rho$, and all these operations for computing each bit of $\alpha$ can be done by  uniform families of polylogarithmic-depth circuits. Therefore, for any $\alpha$ in the considered interval we can recompute $2^n \cdot (1+2^{p(n)}/\rho)^2 + \frac{\alpha-1}{2^{\eta}}+n$ bits of the quantity
$$\sum_{o \in O}{2^{\log{\Pr[C(x)=o]}\cdot \alpha - \log{\Pr[C(x')=o]}\cdot(\alpha-1)}}$$
and check whether it is less than $2^{\rho \alpha (\alpha-1)}$ or greater than $2^{\rho \alpha (\alpha-1) + \frac{\alpha-1}{2^{\eta}}}$. As was shown in Theorem \ref{pspace-Renyi-dp-algo} all these operations are computable in polynomial space. Hence, using the algorithm from Theorem \ref{pspace-Renyi-dp-algo} for every $\alpha$ we can verify whether $(C,\rho,\eta)$ is a yes-instance of {\sc Gap-Concentrated-DP} or whether it is a no-instance.

As by Lemma~\ref{lemma:conc-DP-bounded-alpha} we can consider values of $\alpha$ discretized to $\frac{2^{-\eta-1}}{\rho}$ we get that there are at most $\frac{\rho}{2^{-\eta-1}} \cdot (1 + \frac{2^{p(n)}}{\rho})$ values of $\alpha$ to consider, and each such value has representation of  polynomial length. Therefore, we run the algorithm for checking $(\alpha,\alpha\rho)$-R\'enyi-DP at most $(p(n) + \rho)\cdot 2^{\eta+1}$ times. As iterating over exponentially many elements keeps us in $\PSPACE$ if we re-use the space on each iteration, we get a $\PSPACE$-algorithm for checking whether $C$ is $\rho$-concentrated differentially private.
\end{proof}

\begin{proof}[Proof of Theorem~\ref{th:CDP-pspace-hard}]
Fix two dyadic rational numbers $\eps, \delta \in (0,1)$. Let $\eta$ and $\rho$ be positive numbers with finite binary representations such that 
$$0 < \rho + 2 \sqrt{\left(\rho + \frac{1}{2^\eta}\right)\log{(1/\delta)}} + \frac{1}{2^{\eta}} < \eps.$$ 
We show a reduction from {\sc Distinguish $(\eps,\delta)$-DP} to {\sc Gap-Concentrated-DP} problem.
We map each instance $C$ of {\sc Distinguish $(\eps,\delta)$-DP} into an instance $(C,\rho, \eta)$.

As in the proof of Theorem~\ref{th:RDP-pspace-hard} when $C$ is a yes-instance of {\sc Distinguish $(\eps,\delta)$-DP}, $C$ is also $\rho$-CDP. Therefore, $(C,\rho, \eta)$ is a yes-instance of {\sc Gap-Concentrated-DP}. 

In the case where $C$ is a no-instance of {\sc Distinguish $(\eps,\delta)$-DP} we need to show that then $(C,\rho, \eta)$ is a no-instance of {\sc Gap-Concentrated-DP}.
Assume the opposite: for all neighboring inputs $x,x'$, and for all $\alpha > 1$ it holds $D_{\alpha }(C(x)\|C(x')) \leq \rho\alpha + \frac{1}{2^\eta}$. 
Then, for $\rho'=\rho+\frac{1}{2^\eta}$, $C$ is $(\alpha,\rho'\alpha)$-RDP for all $\alpha > 1$, as 
$$D_{\alpha }(C(x)\|C(x')) \leq \rho\alpha + \frac{1}{2^\eta} \leq \left(\rho+\frac{1}{\alpha 2^\eta}\right)\alpha \leq \left(\rho+\frac{1}{2^\eta}\right)\alpha.$$
Therefore, $C$ is $(\rho')$-CDP.
Hence by Theorem~\ref{th:cdp-implies-approx-dp}, $C$ is $(\rho' + 2 \sqrt{\rho' \log{(1/\delta)}},\delta)$-DP. But by our choice of the parameters, $\rho' + 2 \sqrt{\rho' \log{(1/\delta)}} = \rho \alpha + 2 \sqrt{(\rho + \frac{1}{2^\eta})\log{(1/\delta)}} + \frac{1}{2^{\eta}} < \eps$, hence $C$ is $(\eps,\delta)$-DP. But that contradicts our assumption that $C$ is a no-instance of {\sc Distinguish $(\eps,\delta)$-DP}. Therefore, $(C,\rho, \eta)$ is a no instance of {\sc Gap-Concentrated-DP}.

We showed that the reduction is correct. It runs in linear time because the parameters $\rho, \eta$ are computed from $\eps, \delta$ independent of the instance. 

Since Theorem~\ref{lossy_imply_approx_hardness} showed that {\sc Distinguish $(\eps,\delta)$-DP} is $\PSPACE$-hard, we get that {\sc Gap-Concentrated-DP} is $\PSPACE$-hard.
\end{proof}

\begin{proof}[Proof of Theorem~\ref{pspace-trunc-concentrated-algo}] The algorithm is the same as in the proof of Theorem \ref{pspace-concentrated-dp-algo}, with only one change: instead of considering  $\alpha$ in the range $(1,1 + \log{m/\rho})$, where $m$ depends on a discretization of probabilities in the Markov chain of the program $C$, we consider the range $(1,\omega)$. Although the parameter $\omega > 1 + \log{m/\rho}$, we again need to iterate over at most exponentially many in the length of the input values of $\alpha$. Therefore, getting a $\PSPACE$-algorithm for {\sc Gap-Truncated-Concentrated-DP}.
\end{proof}

\begin{proof}[Proof of Theorem~\ref{th:TCDP-pspace-hard}]
Fix two dyadic rational numbers $\eps, \delta \in (0,1)$. Let $\eta$, $\omega$, and $\rho$ be positive numbers with finite binary representation, such that 
$$0 < \rho + 2 \sqrt{(\rho + \frac{1}{2^\eta})\log{(1/\delta)}} + \frac{1}{2^{\eta}} < \eps, \text{ and}$$
$$\log{1/\delta} < (\omega-1)^2\rho.$$

Similarly to the proof of Theorem~\ref{th:CDP-pspace-hard} we show a reduction from {\sc Distinguish $(\eps,\delta)$-DP} to {\sc Gap-Concentrated-DP} problem.
We map each instance $C$ of {\sc Distinguish $(\eps,\delta)$-DP} into an instance $(C,\rho, \omega, \eta)$.

As in the proof of Theorem~\ref{th:RDP-pspace-hard} and Theorem~\ref{th:CDP-pspace-hard} when $C$ is a yes-instance of {\sc Distinguish $(\eps,\delta)$-DP}, $C$ is also $(\rho)$-CDP. Therefore, $(C,\rho, \eta)$ is a yes-instance of {\sc Gap-Concentrated-Truncated-DP}. 

In case of $C$ is a no-instance of {\sc Distinguish $(\eps,\delta)$-DP} we need to show that then $(C,\rho, \omega, \eta)$ is a no-instance of {\sc Gap-Concentrated-DP}.
By way of contraposition, suppose that for all neighboring inputs $x,x'$, and for all $1 < \alpha < \omega$ holds $D_{\alpha }(C(x)\|C(x')) \leq \rho\alpha + \frac{1}{2^\eta}$. 
Then, for $\rho'=\rho+\frac{1}{2^\eta}$, $C$ is $(\alpha,\rho'\alpha)$-RDP for all $\alpha > 1$, as 
$$D_{\alpha }(C(x)\|C(x')) \leq \rho\alpha + \frac{1}{2^\eta} \leq (\rho+\frac{1}{\alpha 2^\eta})\alpha \leq (\rho+\frac{1}{2^\eta})\alpha.$$
Therefore, $C$ is $(\rho',\omega)$-TCDP.
Hence by Theorem~\ref{th:tcdp-implies-approx-dp}, $C$ is $(\rho' + 2 \sqrt{(\rho' \log{(1/\delta)}},\delta)$-DP. But by our choice of the parameters, $\rho' + 2 \sqrt{(\rho' \log{(1/\delta)}} = \rho \alpha + 2 \sqrt{(\rho + \frac{1}{2^\eta})\log{(1/\delta)}} + \frac{1}{2^{\eta}} < \eps$, hence $C$ is $(\eps,\delta)$-DP. This means that $P$ is not a no-instance of {\sc Distinguish $(\eps,\delta)$-DP}.

We showed that the reduction is correct. As the parameters $\rho,$ $\omega,$ and $\eta$ depend only on the constants $eps, \delta$, overall we get a linear-time deterministic algorithm for such reduction. Hence, {\sc Gap-Truncated-Concentrated-DP} is $\PSPACE$-hard.
\end{proof}

\end{document}